\documentclass[11pt,letterpaper]{article}
\usepackage{thmtools}
\usepackage{thm-restate}
\usepackage{bm}
\usepackage{bbm}
\usepackage{mathrsfs}           %
\usepackage{vmargin,fancyhdr}   %
\usepackage{tikz}
\usetikzlibrary{positioning,chains,fit,shapes,calc}
\usepackage{algorithm}
\usepackage{algpseudocode}

\usepackage{subcaption}
\usepackage{enumerate}

\usepackage{amsmath,amssymb, amsthm}    %
\usepackage{verbatim}           %
\usepackage{xspace}             %
\usepackage{graphicx,float}     %
\usepackage{ifthen,calc}        %
\usepackage{textcomp}           %
\usepackage{fancybox}           %
\usepackage{hhline}             %
\usepackage{float}              %

\usepackage[vflt]{floatflt}     %
\usepackage[small,compact]{titlesec}
\usepackage{setspace}
\usepackage{color}
\definecolor{ForestGreen}{rgb}{0.1333,0.5451,0.1333}
\usepackage{thumbpdf}
\usepackage[letterpaper,
colorlinks,linkcolor=ForestGreen,citecolor=ForestGreen,
backref,
bookmarks,bookmarksopen,bookmarksnumbered]
{hyperref}

\setpapersize{USletter}
\setmarginsrb{1in}{.5in}        %
{1in}{.5in}        %
{.25in}{.25in}     %
{.25in}{.5in}      %
\newcommand{\showccc}[0]{0}
\newcommand{\ccc}[2][nothing]{%
	\ifthenelse{\showccc=0}{}{
		\ensuremath{^{\Lsh\Rsh}}\marginpar{\raggedright\tiny\textsf{%
				\ifthenelse{\equal{#1}{nothing}}{}{\textbf{#1}\\}#2}}}}
\newcounter{hours}\newcounter{minutes}
\newcommand{\hhmm}{%
	\setcounter{hours}{\time/60}%
	\setcounter{minutes}{\time-\value{hours}*60}%
	\ifthenelse{\value{hours}<10}{0}{}\thehours:%
	\ifthenelse{\value{minutes}<10}{0}{}\theminutes}
\ifthenelse{\showccc=0}{\rhead{}}{\rhead{\today \ [\hhmm]}}
\newtheorem{theorem}{Theorem}[section]

\newtheorem{definition}[theorem]{Definition}

\newtheorem{lemma}[theorem]{Lemma}
\newtheorem{fact}[theorem]{Fact}

\newtheorem{assumption}[theorem]{Assumption}

\usepackage{float}
\floatstyle{plain}\newfloat{myfig}{t}{figs}[section]
\floatname{myfig}{\textsc{Figure}}
\floatstyle{plain}\newfloat{myalg}{H}{algs}[section]
\floatname{myalg}{}
\newcommand{\defeq}{:=}
\newcommand{\nnz}{\textup{nnz}}
\newcommand{\norm}[1]{\left\lVert#1\right\rVert}
\newcommand{\inprod}[2]{\left\langle#1, #2\right\rangle}

\newcommand{\R}[0]{\mathbb{R}}

\newcommand{\N}[0]{\mathbb{N}}
\newcommand{\diag}[1]{\textbf{\textup{diag}}(#1)}
\newcommand{\half}[0]{\frac{1}{2}}
\newcommand{\1}[0]{\mathbf{1}}

\newcommand{\smax}[0]{\textup{smax}}
\newcommand{\smin}[0]{\textup{smin}}
\newcommand{\Tr}[0]{\textup{Tr}}
\newcommand{\Oh}[1]{O\left(#1\right)}
\newcommand{\tOh}[1]{\tilde{O}\left(#1\right)}
\newcommand{\eps}[0]{\epsilon}

\newcommand{\Sym}[0]{\mathbb{S}}
\newcommand{\PSDSet}[0]{\mathbb{S}_{\geq 0}}
\newcommand{\lmax}[1]{\lambda_{\max}\left(#1\right)}
\newcommand{\lmin}[1]{\lambda_{\min}\left(#1\right)}
\newcommand{\lte}[1]{\log \Tr \exp\left(#1\right)}
\newcommand{\te}[1]{\Tr \exp\left(#1\right)}

\newcommand{\trprod}[2]{\inprod{#1}{#2}}

\newcommand{\Cx}{\sum_{i=1}^d x_i C_i}
\newcommand{\Px}{\sum_{i=1}^d x_i P_i}
\newcommand{\cj}{C_{j:}}
\newcommand{\pj}{P_{j:}}
\newcommand{\ci}{C_{:i}}
\newcommand{\pii}{P_{:i}}
\newcommand{\xnext}{x^{(t+1)}}
\newcommand{\xcurr}{x^{(t)}}
\newcommand{\pgradi}{\nabla_{P_i}^{(t)}}
\newcommand{\pGrad}{\mathbf{\nabla}_{\mathbf{P}}^{(t)}}
\newcommand{\pGradi}{\mathbf{\nabla}_{P_i}^{(t)}}
\newcommand{\cGrad}{\mathbf{\nabla}_{\mathbf{C}}^{(t)}}
\newcommand{\cGradi}{\mathbf{\nabla}_{C_i}^{(t)}}
\newcommand{\cgradi}{\nabla_{C_i}^{(t)}}
\newcommand{\pgradplaini}{\nabla_{P_i}}
 
\newcommand{\cgradplaini}{\nabla_{C_i}}
\newcommand{\pGradplain}{{\mathbf{\nabla}}_{\mathbf{P}}}
\newcommand{\cGradplain}{{\mathbf{\nabla}}_{\mathbf{C}}}
\newcommand{\pGradplaini}{\mathbf{\nabla}_{\mathbf{P}_i}}
\newcommand{\cGradplaini}{\mathbf{\nabla}_{\mathbf{C}_i}}
\newcommand{\pGradiniti}{\mathbf{\nabla}_{P_i}^{(1)}}
\newcommand{\cGradiniti}{\mathbf{\nabla}_{C_i}^{(1)}}

\newcommand{\allP}{\textbf{P}}
\newcommand{\allC}{\textbf{C}}
\newcommand{\Pxcurr}{\sum_{i=1}^d x_i^{(t)} P_i}
 
\newcommand{\Cxcurr}{\sum_{i=1}^d x_i^{(t)} C_i}

\newcommand{\packstep}{A}
\newcommand{\covstep}{B}
\newcommand{\rcovstep}{\tilde{\covstep}}
\newcommand{\packsq}{G}
\newcommand{\covsq}{H}
\newcommand{\rcovsq}{\tilde{\covsq}}
\usepackage{pxfonts}
\newcommand{\sml}{\varsigmaup}

\newcommand{\tpexp}{\mathcal{T}^{\textup{pack}}_{\textup{exp}}}
\newcommand{\tcexp}{\mathcal{T}^{\textup{cov}}_{\textup{exp}}}
\definecolor{forestgreen}{rgb}{0.133, 0.545, 0.133}

\definecolor{burntorange}{rgb}{0.8, 0.33, 0.0}

\newcommand\numberthis{\addtocounter{equation}{1}\tag{\theequation}}

\newcommand{\inp}{I_{n_p}}
  \usepackage{nth}
  \usepackage{intcalc}

\usepackage{url}

\begin{document}

There is an error in this manuscript. In particular, the analysis bounding the number of iterations of our algorithm in Section~\ref{sec:cleanup} is not correct. The error is in the proofs of Lemmas~\ref{lem:numfastiters} and~\ref{lem:numSlowIters}, wherein we claim that the quantities $\inprod{P_i}{Y^{(t)}}$ and $\inprod{C_i}{Z^{(t)}}$ are monotone, which is true in the scalar case but not necessarily true in the matrix case. To the best of our knowledge, the remainder of the paper is correct as stated (the potential analysis in Sections~\ref{sec:lpwarmup} and~\ref{sec:sdpfull}, and approximation tolerance analysis in Section~\ref{sec:cleanup}). While we believe it is straightforward to modify our analysis to show the algorithm terminates in $O(d \cdot \text{poly}(\log(nd\rho), \eps^{-1}))$ iterations, it is currently unclear to us how to modify our analysis in Section~\ref{sec:cleanup} to obtain a $O(\text{poly}(\log(nd\rho), \eps^{-1}))$ bound on the number of iterations. We believe this remains an important open problem in positive semidefinite programming.

\newpage
	\begin{titlepage}
		\def\thepage{}
		\thispagestyle{empty}
		
		\title{Positive Semidefinite Programming: Mixed, Parallel, and Width-Independent} 
		
		\date{}
		\author{
			Arun Jambulapati \\
			Stanford University \\
			{\tt jmblpati@stanford.edu} 
			\and
			Yin Tat Lee \\
			University of Washington \\
			{\tt yintat@uw.edu} 
			\and
			Jerry Li \\
			Microsoft Research \\
			{\tt jerrl@microsoft.com} 
			\and
			Swati Padmanabhan \\
			University of Washington \\
			{\tt pswati@uw.edu} 
			\and
			Kevin Tian\thanks{Part of this work was done when KT was visiting the University of Washington.} \\
			Stanford University \\
			{\tt kjtian@stanford.edu}
		}
		
		\maketitle

\abstract{
We give the first approximation algorithm for mixed packing and covering semidefinite programs (SDPs) with polylogarithmic dependence on width. Mixed packing and covering SDPs constitute a fundamental algorithmic primitive with recent applications in combinatorial optimization, robust learning, and quantum complexity. The current approximate solvers for positive semidefinite programming can handle only pure packing instances, and technical hurdles prevent their generalization to a wider class of positive instances. For a given multiplicative accuracy of $\epsilon$, our algorithm takes $\Oh{\log^3(nd\rho) \cdot \epsilon^{-3}}$ parallelizable iterations, where $n$, $d$ are dimensions of the problem and $\rho$ is a width parameter of the instance, generalizing or improving all previous parallel algorithms in the positive linear and semidefinite programming literature. When specialized to pure packing SDPs, our algorithm's iteration complexity is $\Oh{\log^2 (nd) \cdot \epsilon^{-2}}$, a slight improvement and derandomization of the state-of-the-art \cite{Allen-ZhuLO16, PengTZ16, WangMMR15}. For a wide variety of structured instances commonly found in applications, the iterations of our algorithm run in nearly-linear time.

In doing so, we give matrix analytic techniques for overcoming obstacles that have stymied prior approaches to this open problem, as stated in past works \cite{PengTZ16, MahoneyRWZ16}. Crucial to our analysis are a simplification of existing algorithms for mixed positive linear programs, achieved by removing an asymmetry caused by modifying covering constraints, and a suite of matrix inequalities whose proofs are based on analyzing the Schur complements of matrices in a higher dimension. We hope that both our algorithm and techniques open the door to improved solvers for positive semidefinite programming, as well as its applications.
}
 		
	\end{titlepage}

\section{Introduction}
\label{sec:intro}

Efficient solvers for semidefinite programming, a generalization of linear programming, are a fundamental algorithmic tool with numerous applications \cite{VandenbergheB96}. In practical settings where the dimension of the problem is large, polynomial-time algorithms with superlinear dependence on the size of the input, including interior point methods \cite{NesterovN89} and cutting plane methods \cite{KrishnanM06, LeeSW15}, may be prohibitively expensive. As a result, many recent efforts have been made towards developing approximation algorithms. For given accuracy tolerance $\eps$, an approximation algorithm outputs a solution approximating the optimal solution within a factor of $1 + \eps$, for some definition of approximation. These algorithms have runtimes whose dependence on $\eps^{-1}$ is polynomial rather than logarithmic, a bottleneck in their use for extremely precise values of $\eps$, e.g. when $\eps$ is inverse-polynomial in the dimension. However, their per-iteration complexity typically scales \emph{linearly} in the size of the input, so for instances with moderately large tolerance $\eps$, approximation algorithms are desirable.

\paragraph{Pure packing and covering semidefinite programs.}
A family of semidefinite programs (SDPs) that has garnered significant attention from the computer science and operations research communities is \emph{positive} semidefinite programs, whose constraints have additional positivity structure. One well-studied structured family, which we term ``pure packing SDPs'', has primal form\footnote{We use $\PSDSet^n$ to denote the set of $n \times n$ positive semidefinite matrices and $\inprod{\cdot}{\cdot}$ to denote the trace product; please refer to Section~\ref{ssec:notation} for the notation used in this paper.}
\begin{equation}
\begin{aligned}
\label{eq:packingprimal}
\min \inprod{C}{X} \text{ subject to } X \in \PSDSet^n, \inprod{A_i}{X} \geq b_i \text{ for all } i \in [d] \\
\text{for } \{A_i\}_{i \in [d]} \in \PSDSet^n, C \in \PSDSet^n, b \geq 0,
\end{aligned}
\end{equation}
with equivalent dual form
\begin{align}
\label{eq:packingdual}
\max \inprod{b}{x} \text{ subject to } x \geq 0, \sum_{i \in [d]} x_i A_i \preceq C.
\end{align}
Intuitively, this problem can be thought of ``packing'' as many of the $\{A_i\}$ matrices as possible into the constraint matrix $C$, as specified by gains $b$. Another family of structured instances, which we term ``pure covering SDPs'', has primal form 
\begin{equation}
\begin{aligned}
\label{eq:coveringprimal}
\max \inprod{C}{X} \text{ subject to } X \in \PSDSet^n, \inprod{A_i}{X} \leq b_i \text{ for all } i \in [d],\\
\text{for } \{A_i\}_{i \in [d]} \in \PSDSet^n, C \in \PSDSet^n, b \geq 0,
\end{aligned}
\end{equation}
with equivalent dual form
\begin{align}
\label{eq:coveringdual}
\min \inprod{b}{x} \text{ subject to } x \geq 0, \sum_{i \in [d]} x_i A_i \succeq C.
\end{align}
Again, the intuition is that we wish to use as few of the matrices $\{A_i\}$ as possible to meet the constraint. Because of the additional monotonicity structure afforded by these instances, algorithms with improved dependence on problem parameters have been developed. In particular, when the constraint matrices have been appropriately normalized so that $C = I$ (the identity) and the notion of approximation is \emph{multiplicative}, the state-of-the-art approximate SDP solvers for solving packing problems with the form \eqref{eq:packingprimal}, \eqref{eq:packingdual} have a logarithmic dependence on the natural width parameter of the instance, $\rho \defeq \max_i \lmax{A_i}$ \cite{JainY11, PengTZ16, Allen-ZhuLO16}.

For generic semidefinite programs, approximation algorithms without a polynomial dependence on $\rho$ are unknown. Nonetheless, the SDP instances in many settings, including classical approximation algorithms for problems such as maximum and sparsest cut \cite{GoemansW95, AroraRV09}, quantum complexity theory \cite{JainJUW11}, and robust learning and estimation \cite{ChengG18, ChengDG19} may be encoded as positive SDPs, leading to their recent intensive study as an algorithmic primitive.

The flexibility of known approximate positive SDP solvers, in terms of the types of problem instances they can handle, does not match that of the best-known positive linear programming (LP) solvers. In particular, the solvers of \cite{JainY11, PengTZ16, Allen-ZhuLO16}, which achieve polynomial dependence on $\eps^{-1}$ and logarithmic dependence on width, can solve only pure packing instances and do not solve instances of the forms \eqref{eq:coveringprimal}, \eqref{eq:coveringdual}, or more general types of positivity constraints. 

Part of the reason for the ambiguity in nomenclature in the literature in defining ``packing'' and ``covering'' SDPs is as alluded to by \cite{PengTZ16}: because the dual of a matrix-valued SDP is vector-valued, there are many forms a positive SDP can take. Indeed, there are significant technical barriers to generalizing existing algorithms in the literature to handle even pure covering SDP instances (cf. Section \ref{ssec:technical}); the family of pure covering SDP instances was shown in \cite{IyengarPS11} to capture the semidefinite relaxations of maximum cut \cite{GoemansW95}, coloring \cite{KargerMS98}, and Shannon entropy \cite{Lovasz79}. This gap between known algorithms for positive LP solvers and existing SDP solvers leads us to consider a wider family of positive SDP instances.

\paragraph{Mixed positive linear programs.}
In the case of linear programming, the family of positive LPs specialized algorithms can handle is quite broad: in their most general form \cite{Young01, Young14, MahoneyRWZ16}, approximate positive linear programming algorithms return a scalar $\mu$ and vector $x$ corresponding to the problem
\begin{equation}
\label{eq:mixedlp}
\min \mu \text{ subject to } Px \leq \mu p, Cx \geq c, x \geq 0  \text{ for entrywise nonnegative } P, C, p, c
\end{equation}
such that $\mu$ is within a $1 + \epsilon$ multiplicative factor of optimal, and $x$ is feasible. The state-of-the-art runtime depends logarithmically on dimensions of the problem, linearly on the sparsity $\nnz(P) + \nnz(C)$, and polynomially on $\epsilon^{-1}$. This family of problems is large enough to capture pure packing and covering instances: for example, the packing problem
$$\max_{x \geq 0} c^\top x \text{ subject to } Px \leq p$$
has value $\mu^{-1}$, where $\mu$ solves
$$\min \mu \geq 0 \text{ subject to } Px \leq \mu p, c^\top x \geq 1, x \geq 0.$$
The full power of mixed positive linear programs has been used to obtain improved algorithms in settings such as combinatorial optimization, resource allocation, and tomography \cite{Young01, Young14}. 

\paragraph{Our contribution.}
We consider the general problem of designing a width-free approximation algorithm to solve mixed positive SDPs, an independently interesting open problem as posed in \cite{MahoneyRWZ16, PengTZ16}, which has seen recent applications in e.g. spectral graph theory \cite{LeeS17, JambulapatiSS18}. As stated in \cite{MahoneyRWZ16}, which developed the state-of-the-art parallel mixed positive LP solver, this problem is notoriously challenging due to a variety of complications arising from generalizing vector-based tools to matrices. This problem class contains pure packing and covering SDPs of any of the forms \eqref{eq:packingprimal}, \eqref{eq:packingdual}, \eqref{eq:coveringprimal}, and \eqref{eq:coveringdual} by setting one of the constraint sets to be $1 \times 1$ matrices, as well all mixed positive linear programs \eqref{eq:mixedlp}, as special cases. The types of problem instances considered by mixed positive SDP solvers have the form 
\begin{equation}\label{eq:possdpdef}\begin{aligned}\min \mu \text{ subject to } \sum_{i \in [d]} P_i x_i \preceq \mu I_{n_p}, \sum_{i \in [d]} C_i x_i \succeq I_{n_c}, x \geq 0 , \\
\text{for sets of matrices } \left\{P_i \in \PSDSet^{n_p}\right\}_{i \in [d]}, \left\{C_i \in \PSDSet^{n_c}\right\}_{i \in [d]}.\end{aligned}\end{equation}
In this work, we completely resolve this open problem for many instances of mixed positive SDPs, which includes generic packing with commutative (e.g. diagonal) covering constraints\footnote{This was the family of instances considered in \cite{JainY12}, the first proposed width-independent positive SDP solver, and includes the SDP relaxation of maximum cut.}, and we resolve it for all mixed positive SDP instances up to a (mild) polylogarithmic dependence on a natural width measure of the problem\footnote{We call our algorithm ``width-independent'' nonetheless, to be consistent with the existing literature on positive LPs and SDPs with polylogarithmic dependence on width.}. For pure packing instances, our algorithm's runtime matches or improves upon the state-of-the-art claimed runtimes, and its iteration complexity matches that of the current best mixed positive LP solver \cite{MahoneyRWZ16}. 

Perhaps interestingly, our algorithm is a simplification of the natural generalization of the \cite{MahoneyRWZ16} algorithm to matrices. However, as we outline in Section~\ref{ssec:technical}, the simplification we perform brings with it technical challenges in the analysis that were not well-understood even in the positive LP setting, as seen by the presence of an asymmetric pruning operation in \cite{Young01, Young14, MahoneyRWZ16}. We believe the additional matrix-analytic techniques our work utilizes, including a proof method for matrix inequalities based on taking the Schur complement of a matrix in a higher-dimensional space, may be of independent interest of the community. More broadly, we hope that both the algorithmic primitives and conceptual tools introduced in this paper may pave the way towards more improvements to positive LP and SDP solvers, as we discuss in Section~\ref{sec:futurework}.

\subsection{Related work}

The space of recent approximation algorithms for general linear and semidefinite programming \cite{GrigoriadisK95, Nemirovski04, WarmuthK06, AroraK07, CarmonJST19} and their positive variants \cite{LubyN93, PlotkinST95, Young01, JainY11, ZhuO15A, ZhuO15B, WangMMR15, Allen-ZhuLO16, PengTZ16} is vast. We refer the reader to \cite{AroraHK12, ZhuO15B} for excellent surveys of these problems and highlight the works most closely related to ours. 

\paragraph{Mixed positive linear programming.}
The most direct predecessor of our work is that by Young \cite{Young01}, which obtained a parallelizable $\tOh{\eps^{-4}}$\footnote{As is standard in the literature, we use the $\tilde{O}$ notation to hide polylogarithmic factors in problem parameters.}-iteration algorithm for solving problem \eqref{eq:mixedlp} to a $(1 + \eps)$-multiplicative approximation factor\footnote{We refer to algorithms whose iterations consist of matrix-vector products and vector operations as ``parallelizable''. The state-of-the-art sequential algorithm for (pure packing or covering) positive LPs is the breakthrough work of \cite{ZhuO15B}, which performs $\tOh{\nnz/\eps}$ work, where $\nnz$ is the total number of nonzero entries in the constraints.}. Young's algorithm was based on a Lagrangian-relaxation framework, a variation of the entropy-regularization or mirror descent frameworks core to many positive LP solvers. First, Young reduced the problem to deciding, for suitably rescaled nonnegative matrices $P, C$, which one of the following statements holds. 
\begin{itemize}
	\item There exists $x \geq 0$ satisfying $\max_{j \in [n_p]} [Px]_j \leq (1 + \eps)\min_{j \in [n_c]} [Cx]_j$.
	\item There exists no $x \geq 0$ satisfying $Cx \geq \1, Px \leq (1 - \eps)\1$.
\end{itemize} 
Then, Young designed an iterative method to resolve this decision problem by maintaining a vector $x$ so that the potential function
\begin{equation}
\label{eq:lppotential}
\log \sum_{j \in [n_p]} \exp\left(\left[Px\right]_j\right) + \log \sum_{j \in [n_c]} \exp\left(-\left[Cx\right]_j\right),
\end{equation}
a close approximation to the true value of $\max_j [Px]_j - \min_j [Cx]_j$, is nonincreasing, but sufficient ``progress'' is made in updating the size of $x$. In particular, both the potential function and the discrepancy between the potential and the true value of $\max_j [Px]_j - \min_j [Cx]_j$ are bounded by $O(\log n)$ for $n = \max\{n_p, n_c\}$, so once the true value is on the order of $K = O(\log n / \eps)$, the approximation factor is $1 + \eps$. The algorithm terminates either when any coordinate of $Px$ reaches this scale $K$, or the instance is concluded infeasible.

Inspired by approaches based on first-order optimization frameworks \cite{ZhuO15A, ZhuO15B, WangMMR15}, the work of \cite{MahoneyRWZ16} improved this algorithm by taking non-uniform steps based on gradients of the potential function. They showed that their algorithm ran in $\tOh{\eps^{-3}}$ iterations, and that when specialized to pure packing or covering, this could be improved to $\tOh{\eps^{-2}}$, constituting a simplification and derandomization of the prior best parallel runtime in the literature \cite{WangMMR15}. Generalizing these techniques to the SDP setting was stated as an open problem in \cite{MahoneyRWZ16}.

\paragraph{Positive semidefinite programming.}
The first algorithm with a width-independent runtime for positive SDPs was by the work of Jain and Yao \cite{JainY11} for approximating the solution to the pure packing instance \eqref{eq:packingprimal}, \eqref{eq:packingdual}, achieved via appropriate modifications to the positive LP solver of \cite{LubyN93}. In order to deal with the difficulties of adapting the algorithm to potentially non-commuting matrices, Jain and Yao utilized many highly technical facts about the interplay of the projection matrices onto the small eigenspace of iterates. Ultimately, while achieving width-independence constituted a significant milestone, the algorithm obtained had an iteration count of roughly $O(\log^{14} n / \eps^{13})$, each running in superlinear time, motivating the search for SDP solvers with runtimes more comparable to their LP counterparts.

Recent work by \cite{Allen-ZhuLO16, PengTZ16} developed parallelizable positive SDP solvers for pure packing SDPs \eqref{eq:packingprimal}, \eqref{eq:packingdual} with iteration count $\tOh{\eps^{-3}}$. These algorithms have iterations that are much less expensive than those of \cite{JainY11} (and implementable in nearly-linear time for constant $\eps$) and are similarly based on performing appropriate modifications to positive LP solvers \cite{ZhuO15A, Young01}. Moreover, both algorithms are more directly analyzed through the perspective of continuous optimization, leading to significantly simpler analyses. Finally, up to an additional logarithmic term in the runtime, both of these works claimed a runtime improvement to $\tOh{\eps^{-2}}$ iterations using randomized dynamic bucketing techniques that originated in the LP literature \cite{WangMMR15}.

As previously discussed, existing positive SDP solvers are capable of handling only instances of the form \eqref{eq:packingprimal}, \eqref{eq:packingdual}; the problem of extending these techniques to the more general family of positive SDP instances we consider was also explicitly stated in \cite{PengTZ16}. While the difficulty in generalizing these solvers in a similar way as has been achieved in the LP literature is not immediately clear, several challenges arise due to monotonicity statements that are true in the vector case, but false in the matrix case. Indeed, attempts to make these generalizations have had a somewhat notorious history of breaking due to mistaken claims of monotonicity properties \cite{Zhu19}, including a prior version of \cite{PengTZ16}, and a proposed algorithm for resolving the mixed positive SDP problem \cite{JainY12}. For a more extensive discussion on the history of this problem, we refer the reader to \cite{Allen-ZhuLO16}.

\subsection{Technical challenges}
\label{ssec:technical}

We now describe some of the key technical challenges that barred generalizations of prior work from solving the more general problems considered in this paper.

\paragraph{Asymmetry and explicit pruning.}
All existing solvers for mixed packing and covering LPs \cite{Young01, Young14, MahoneyRWZ16} involve an explicit pruning step, where satisfied covering constraints $\cj$, i.e. constraints that have reached the termination threshold $\cj^\top x \geq K$, are removed from the covering matrix. The intuition is that such constraints are no longer ``actively relevant to the problem'', so it is safe to ignore them. 

The reason for this asymmetry is that crucial to showing that the potential function \eqref{eq:lppotential} is nonincreasing is considering a second-order Taylor expansion of each summand. Since the increments on $x$ are chosen to have magnitude bounded by a multiplicative $O(K^{-1})$, and the algorithm terminates if any coordinate of $Px$ exceeds $K$, the update to the packing potential is in the range in which a second-order upper bound to the exponential holds. However, the covering potential is made up of the coordinates of $Cx$, which may be arbitrarily large while at least one coordinate is not covered, and so it is not clear how to use a similar Taylor approximation without explicit pruning. 

Dropping constraints in the SDP setting corresponds to removing the span of eigenvectors above an eigenvalue threshold. This is complicated since noncommutativity may cause basis vectors to change from iteration to iteration, and it is also unclear how to efficiently implement this projection, leading to the worse iteration complexity of \cite{JainY11}. Indeed, a similar asymmetry issue is present even in many of the pure LP solvers, preventing their direct use for pure covering LPs and potentially impeding progress towards an accelerated parallel LP solver (cf. Section~\ref{sec:futurework}).

\paragraph{Incomparability of eigenspaces.}
Even if the explicit pruning of the large eigenspace were implementable, directly reasoning about the change in the potential function from a vector $x$ to a vector $x \circ (1 + \delta)$ remains difficult, involving bounding changes between the projection matrices onto the small eigenspaces of $\sum_i C_i x_i$ and $\sum_i C_i (x_i + \delta_i)$, for example by bounding the canonical angles between them \cite{PaigeS81}. A similar analysis was performed in \cite{JainY11}; however, doing so without incurring a worse runtime dependence on $\eps^{-1}$ remains daunting.

Prior efforts to resolve these issues \cite{JainY12} attempted to use natural monotonicity claims that are true in the scalar case, but incorrect in the matrix case. For $\Psi = \sum_i C_i x_i$ and $B = \sum_i C_i x_i (1 + \Delta_i)$, where $\Delta_i$ is the multiplicative increment on $x_i$, examples of monotonicity statements that do \textit{not} generally hold for non-commuting matrices include $\exp(\Psi) \preceq \exp(\Psi + B)$ and $P_{\Psi}BP_{\Psi} \preceq B$, where $P_{\Psi}$ is the projection matrix onto the small eigenspace of $\Psi$; these arguments are essential in the analyses of the counterpart LP algorithms, necessitating additional changes. 

\paragraph{Fine-grained notions of matrix inequalities.}
The work of \cite{Allen-ZhuLO16} introduced a continuous perspective in bounding the change of the packing potential between iterates, integrating to track the ``path'' of the potential from $\Psi$ to $\Psi + B$. Crucially, they used an extended Lieb-Thirring inequality (cf. Section~\ref{ssec:facts}) to compare inner products of the form $\inprod{B}{\exp(\Psi)}$ and $\inprod{B}{\exp(\Psi + tB)}$ for any $0 \leq t \leq 1$, leading to their result. The proof of this inequality requires an upper bound on the size of $B$ in order to compare quantities such as $\inprod{B^2}{\exp(-\Psi - tB)}$ and $\inprod{B}{\exp(-\Psi - tB)}$, leading to the inability to handle covering instances (as such a bound on $B$ does not hold in our setting), and first-pass attempts to prove necessary monotonicity properties using this technique fail even for simple $2\times 2$ matrices (cf. Section~\ref{ssec:generalizingIdeas}). 

Moreover, the algorithm of \cite{Allen-ZhuLO16} scales down the step size by an additional factor of $O(\eps)$ to use a first-order Taylor expansion, leading to a worse iteration count, partially resolved using dynamic bucketing \cite{WangMMR15} in the LP case, and later derandomized by \cite{MahoneyRWZ16}. To generalize the improved algorithm of \cite{MahoneyRWZ16} to the SDP setting, considering the second-order Taylor expansion appears to be necessary, which is more difficult in the matrix case.

\subsection{Our approach}

We now give an overview of the techniques developed in this work to overcome the above mentioned technical barriers.

\paragraph{Avoiding the explicit pruning.}
Our first simplification, avoiding the explicit pruning operation in the algorithm, comes from a simple observation that in the LP setting, the covering potential is not affected much by large entries of $Cx$. In particular, because the covering potential
\begin{equation*}
-\log \sum_{j \in [n_c]} \exp\left(-\left[Cx\right]_j\right) \approx \min_{j \in [n_c]} \left[Cx\right]_j
\end{equation*}
has an exponential dropoff in the dependence on any entry of $Cx$ which is much larger than the minimal entry, and we are guaranteed by the algorithm's termination condition that there is at least one small entry, we can simply bound the effect of the large entries on the covering potential separately, and perform a Taylor expansion restricted to the small coordinates.

\paragraph{Avoiding directly reasoning about the eigenspaces.}
The notion of Taylor expanding around the small coordinates in the SDP setting is to Taylor expand around the small eigenspace. However, as previously mentioned, this becomes difficult when one tries to claim a fact such as
\begin{equation*}
-\lte{-\Psi - B}  \geq -\lte{-\Psi - P_{\Psi} B P_{\Psi}},
\end{equation*}
the generalization of the argument appearing in the LP setting, because $P_{\Psi} B P_{\Psi} \preceq B$ is false in general (cf. Section~\ref{ssec:generalizingIdeas}). To circumvent this issue, inspired by \cite{Allen-ZhuLO16}, we adopt a continuous perspective on the change in the potential and control the function $-\lte{-\Psi - tB}$ for each $0 \leq t \leq 1$, by locally bounding its derivative. 

\paragraph{Avoiding black-box use of matrix inequalities.}
A direct analysis of the derivative of our potential by applying the extended Lieb-Thirring inequality fails on a minimal example (cf. Section~\ref{ssec:generalizingIdeas}). However, by carefully interlacing the technique of separately arguing about the small and large eigenspaces with the extended Lieb-Thirring proof, we are able to obtain a tighter bound. 

Scaling down the step size by a factor of $\eps$ and combining with the techniques previously mentioned yields an $\tOh{\eps^{-4}}$-iteration algorithm with logarithmic dependence on the width. However, in order to match the iteration count of state-of-the-art LP solvers, we require a more direct argument about the second-order Taylor expansion. We achieve this by controlling the derivative of a local potential interpolating between the first and second order Taylor expansions. 

Interestingly, both our bounds on the interactions between the small and large eigenspaces, and the second-order expansion term in terms of the first, are achieved through the use of the well-known Schur complement condition for positive semidefiniteness of a matrix. The key in both settings is to encode the desired inequality as taking a Schur complement of a matrix in a larger dimension: see Lemma~\ref{lem:matrixcs}, a specialization of Kadison's inequality, for a simple example. A similar technique results in an alternative proof of the extended Lieb-Thirring inequality, which we include in Appendix~\ref{app:technical} as an additional exposition of our techniques, and we hope that it is of independent interest.

\subsection{Optimization vs. decision, and discussion of width dependence}
\label{ssec:width}
In the full generality of problem instances it is capable of handling, our algorithm has a mild dependence on width parameters of the instance due to two reasons.

Firstly, as in existing work \cite{Young01, Allen-ZhuLO16, MahoneyRWZ16, PengTZ16}, our algorithm for solving the optimization problem \eqref{eq:possdpdef} begins by reducing via binary search to a sequence of decision problems of the following form: for appropriately rescaled $\{P_i\}$, $\{C_i\}$:  
\begin{equation}\label{eq:feasible}\text{give } x \ge 0 \text{ such that }\lmax{\sum_{i \in [d]} x_i P_i} \leq (1 + \epsilon) \lmin{\sum_{i \in [d]} x_i C_i},\end{equation}
or 
\begin{equation}\label{eq:infeasible} \text{demonstrate infeasibility of the constraints }\sum_{i \in [d]} x_i C_i \succeq I_{n_c}, \; \sum_{i \in [d]} x_i P_i \preceq (1 - \epsilon) I_{n_p}, \; x \geq 0.\end{equation}
Given upper and lower bounds $\mu_{\text{upper}}$ and $\mu_{\text{lower}}$ on the optimal $\mu$ in \eqref{eq:possdpdef}, our algorithm for solving \eqref{eq:possdpdef} incurs a multiplicative overhead of roughly $O(\log\log(\mu_{\text{upper}}/\mu_{\text{lower}}) + \log(\eps^{-1}))$ when compared to the algorithm for solving the decision problem \eqref{eq:feasible}, \eqref{eq:infeasible} due to this binary search, as discussed in Appendix~\ref{app:reduction}. The main body of the paper will describe an algorithm for solving the decision problem \eqref{eq:feasible}, \eqref{eq:infeasible} and its runtime, to be consistent with prior work, and we defer details on the implications of this runtime in terms of the original optimization problem \eqref{eq:possdpdef} to Appendix~\ref{app:reduction}.

Secondly, given an instance of the decision problem \eqref{eq:feasible}, \eqref{eq:infeasible}, our algorithm incurs a polylogarithmic dependence on the quantity
\begin{equation}\label{eq:widthdef}\rho \defeq \max_{i \in [d]} \frac{\lmax{C_i}}{\lmax{P_i}},\end{equation}
the largest ratio of maximum eigenvalues amongst the corresponding covering and packing matrices. We note that we may assume $\rho \geq 1$, else \eqref{eq:infeasible} is clearly infeasible. All the current work in the positive linear and semidefinite programming also has logarithmic dependences on width parameters of the instance, such as the largest entry of a packing constraint for packing LPs, or the largest eigenvalue of a packing matrix for packing SDPs. However, in such instances the notion of width could typically be assumed to be bounded by a polynomial in the dimension; to the best of our knowledge, due to difficulties which arise when comparing matrices, such an assumption cannot be made for mixed positive SDPs.

Nonetheless, this issue does not arise for many interesting types of mixed SDPs, including those with commuting covering constraints (such as the dual of the SDP relaxation of maximum cut \cite{LeeP19}). Moreover, for many applications in e.g. spectral sparsification \cite{LeeS17, JambulapatiSS18} or more generally when the covering matrices are multiples of the packing matrices, the width is easily bounded. In prior reduction-based approaches to mixed positive SDPs developed in these works, where the two-sided packing and covering constraints were reduced to a one-sided packing problem, the dependence on this ratio was polynomial; we improve this dependence exponentially generically. It is interesting to explore if this dependence is necessary. We consider designing an algorithm with either an improved dependence on width, or a logarithmic dependence on a notion of width which is polynomially bounded, to be an important open problem.

\subsection{Discussion of iteration complexity}
\label{ssec:hardness}
As is typical of approximate SDP solvers \cite{WarmuthK06, AroraK07, PengTZ16, Allen-ZhuLO16, CarmonDST19}, all steps of our algorithm may be performed in time nearly-linear in the sparsity of the positive SDP instance, except for possibly the bottleneck step of computing the trace product between an exponential of a matrix and another matrix, in order to compute coordinates of gradients of our objective. In general, this primitive requires a full eigendecomposition, requiring time $O(n^{\omega})$, where $\omega \approx 2.37$ is the matrix multiplication constant, and may be parallelized with polylogarithmic depth \cite{JainY11}. 

The bottleneck step is implementable in nearly-linear time when the matrix to be exponentiated has bounded spectrum, by using a technique involving multiplying a Taylor series approximation to the exponential with a small number of random vectors, via the Johnson-Lindenstrauss lemma (a technique introduced in \cite{AroraK07}). Typically, for algorithms approximately solving pure packing SDPs, the bounded spectrum assumption is not restrictive, since the termination condition implies that the spectrum of the maintained matrix will always be bounded, resulting in efficient iterations; this is the case for our specialization to pure packing as well. However, a similar property does not hold for algorithms approximating pure covering SDPs, or positive SDPs in general.

This is unsurprising: given access to a black-box oracle for approximating the value to \eqref{eq:possdpdef}, even using a pure covering instance with a single covering matrix, it is possible to detect if a positive semidefinite matrix has a kernel, a problem which is not known to be solvable faster than matrix multiplication. Indeed, for scalar $x$, being able to approximate the problem
\[\min \lambda \text{ subject to } I_n x \preceq \lambda I_n, Cx \succeq I_n, x \geq 0\]
is equivalent to being able to return a multiplicative approximation on the smallest eigenvalue of $C$, and will be zero if and only if $C$ is singular. Therefore, in general we need to assume more structure in order to implement steps faster than matrix multiplication time.

For many structured instances of positive SDPs, the time to compute a trace product with the matrix exponentials encountered in our algorithm is considerably cheaper than matrix multiplication, which we elaborate on in Section~\ref{ssec:JL}. One example of an assumption that we show is sufficient for implementing steps in nearly-linear time is the following:

\begin{assumption}
	\label{assume:inverse}
	For a positive SDP with covering matrices $\{C_i\}_{i \in [d]}$, assume access to a linear system solver in the matrix $\alpha I + \sum_{i \in [d]} x_i C_i$ for arbitrary nonnegative $\alpha, \{x_i\}_{i \in [d]}$ in time $\sum_{i \in [d]} \nnz(C_i)$.
\end{assumption}

For many applications, including diagonal covering matrices (or more generally, commuting covering matrices in a specified basis), or symmetric, diagonally dominant matrices \cite{SpielmanT14} (e.g. graph Laplacians), Assumption~\ref{assume:inverse} holds. In particular, we note that the former case is the assumed structure made by a prior proposed solver for our problem instance \cite{JainY12} and includes e.g. the SDP relaxation of maximum cut, and the latter case has appeared in various applications of positive SDP solvers \cite{LeeS17, ChengG18, JambulapatiSS18}. Moreover, we show how to remove an assumption used in prior work on pure packing SDP solvers \cite{PengTZ16, Allen-ZhuLO16} assuming access to Cholesky decompositions of the constraint matrices. We find a formal hardness dichotomy on the instances of mixed positive SDPs which admit efficient solvers, in the vein of \cite{KyngZ17, KyngWZ20}, to be an interesting open problem, and leave it for future work (cf. Section~\ref{sec:futurework}).  

\subsection{Our results}
\label{ssec:results}
The main result of our work is the following.

\paragraph{Main claim.} 
There is an algorithm that takes positive SDP instance \eqref{eq:feasible}, \eqref{eq:infeasible} parameterized by 
\[\left\{P_i \in \PSDSet^{n_p}\right\}_{i \in [d]}, \left\{C_i \in \PSDSet^{n_c}\right\}_{i \in [d]},\]
which either returns a vector $x \geq 0$ such that
\[\lmax{\sum_i x_i P_i} \leq (1 + \epsilon) \lmin{\sum_i x_i C_i},\]
or concludes that the constraints\footnote{We assume that the constraints involve the identity matrix, as in prior works. To handle arbitrary constraints, it suffices to scale the constraint matrices appropriately, and define the width for the rescaled matrices. By first computing the summed matrix and then applying the scaling, iterations are still able to capture the original instance's sparsity up to the additive cost of applying the scaling.}
$$\sum_i x_i C_i \succeq I_{n_c}, \; \sum_i x_i P_i \preceq (1 - \epsilon) I_{n_p}, \; x \geq 0$$
are infeasible. Moreover, the vector $x$ certifies that the instance is feasible, e.g. $\sum_i P_i x_i \preceq (1-\epsilon) I_{n_p}, \sum_i C_i x_i \succeq I_{n_c}$ hold. The number of iterations of the algorithm is $O\left(\log^3(nd\rho) \cdot  \eps^{-3} \right)$, where $n \defeq \max\{n_p, n_c\}$, and $\rho$ is defined as in \eqref{eq:widthdef}. For a full discussion of the per-iteration complexity, see Section~\ref{ssec:implementation}; we summarize some guarantees here. If Assumption~\ref{assume:inverse} holds, each iteration is implementable in $O(\log^2(nd\rho) \cdot \eps^{-1})$ parallel depth, and $O(\nnz \cdot \log^2(nd\rho) \cdot \eps^{-3})$ total work, where $\nnz$ is the total nonzeros in the constraint matrices. Each iteration is always implementable in polylogarithmic parallel depth, and $O(\nnz + n^{\omega}) $ total work.

We highlight the fact that for pure packing instances, our iteration complexity is the same as all other positive SDP algorithms, and discuss potential improvements in Section~\ref{sec:futurework}.

\paragraph{Specializations.}
Our algorithm also admits the following improved runtimes for more structured SDP instances, with the same per-iteration costs.
\begin{itemize}
	\item For pure packing SDP instances, the number of iterations is bounded by $O\left(\log^2(nd) \cdot  \eps^{-2} \right)$.
	\item For pure covering SDP instances, the number of iterations is bounded by $O\left(\log^3(nd\rho) \cdot  \eps^{-2} \right)$.
	\item For mixed positive SDP instances with commuting covering constraint matrices, the number of iterations is bounded by $O\left(\log^3(nd) \cdot  \eps^{-3} \right)$, e.g. there is no dependence on the width $\rho$.
\end{itemize}

Our algorithm also captures the best-known parallel positive LP solvers by considering the appropriate diagonal specializations. The organization of our paper is as follows.

\begin{itemize}
	\item \textbf{Section~\ref{sec:prelims}: Preliminaries.} We discuss the notation used in the paper, and state several useful facts from linear algebra and matrix calculus. Then, we give the general form of the algorithm which we will analyze, as well as some basic properties of the algorithm.
	\item \textbf{Section~\ref{sec:lpwarmup}: Linear Program Potential Analysis.} As an exposition of our algorithmic modification of removing explicit pruning of covered constraints, we give a simplified analysis of the potential used in \cite{MahoneyRWZ16} without pruning. We  then discuss the potential issues arising in attempts to generalize these changes to the matrix case, with explicit counterexamples.
	\item \textbf{Section~\ref{sec:sdpfull}: Semidefinite Program Potential Analysis.} We demonstrate how to combine modifications to prior proof techniques with novel matrix analysis tools to analyze the potential function for mixed positive SDPs.
	\item \textbf{Section~\ref{sec:cleanup}: Convergence Analysis.}  We show how to use the potential analysis to show an iteration bound. This convergence analysis is essentially identical to that in \cite{MahoneyRWZ16}, and the technical bottleneck is the potential analysis. We  discuss implementation issues as well and give the final runtime.
	\item \textbf{Section~\ref{sec:futurework}: Future Work.} We discuss potential avenues for improving and extending our algorithm, situating our techniques in the frontier of LP and SDP solvers in general. We also discuss the potential challenges they present.
\end{itemize} 	%
\section{Preliminaries}
\label{sec:prelims}

\subsection{Notation}
\label{ssec:notation}

\paragraph{Vectors.}
Throughout, for $d \in \N$, $[d] \defeq \{i \mid 1 \leq i \leq d, i \in \N\}$. Denote $\R^d_{\geq 0}$ as the set of $d$-dimensional vectors that are entrywise nonnegative. We use $\1_d \in \R^d$ to refer to the all-ones vector. For vectors $x, y$, $x \circ y$ is the entrywise product of $x$ and $y$, and $x^2 \defeq x \circ x$ is the entrywise square. We use $x \leq y$ to mean that entrywise, $x$ is smaller than $y$, and when $c$ is a scalar, $x \leq c$ is also entrywise.

\paragraph{Matrices.}
We use $\Sym^n$ and $\Sym^n_{\geq 0}$ to refer to the set of symmetric $n \times n$ matrices and symmetric positive semidefinite matrices, respectively.  The $n \times n$ identity matrix is denoted by $I_n$. When comparing $A, B \in \Sym^n$, we use the Loewner ordering $\preceq$, e.g. $A \preceq B \implies B - A \in \PSDSet^n$. When clear from context, e.g. $A \succeq 0$, we use $0$ to denote the all-zeroes matrix. For a vector $v \in \R^n$, $\diag{v}$ is the diagonal $n \times n$ matrix with $\diag{v}_{ii} = v_i$. The smallest and largest eigenvalues of $A \in \Sym^n$ are $\lmin{A}$ and $\lmax{A}$, respectively. For $A, B \in \Sym^n$, we define the inner product to be the trace product , $\inprod{A}{B} \defeq \Tr[AB] = \sum_{i, j \in [n]} A_{ij} B_{ij}$. For a set $S \subseteq [n]$ and matrix $M$, we use $M_{SS}$ to denote the submatrix with rows and columns restricted to coordinates in $S$. Our convention is to view $M_{SS}$ as a $|S| \times |S|$-sized matrix. We use $\nnz(A)$ to refer to the number of nonzero entries of a matrix $A$.

\paragraph{Linear programs.}
When treating positive linear programs, we refer to the packing constraint matrix (when applicable) as $P \in \R_{\geq 0}^{n \times d}$, and the covering constraint matrix (when applicable) by $C \in \R_{\geq 0}^{n \times d}$, e.g. $P$ and $C$ both have $n$ constraints. When $P$ and $C$ have $n_p$ and $n_c$ constraints respectively for $n_p \neq n_c$, we define $n \defeq \max\{n_p, n_c\}$. We use $\pj$ and $\pii$, respectively, for the rows and columns of the packing matrix, and $\cj, \ci$ are similarly defined.

\paragraph{Semidefinite programs.}
When treating positive semidefinite programs, we refer to the family of packing constraint matrices by $\allP = \{P_i \mid P_i \in \Sym^{n}\}_{i \in [d]}$, and the family of covering constraint matrices by $\allC = \{C_i \mid C_i \in \Sym^{n}\}_{i \in [d]}$. As in the case for positive linear programs, for constraint matrices of unequal size, we define $n\defeq \max\{n_p, n_c\}$. Finally, for $\allP, \allC$ corresponding to a mixed packing-covering instance, we define the width 
\[\rho \defeq \max_{i \in [d]} \frac{\lmax{C_i}}{\lmax{P_i}}.\]
We use $\nnz(\allP) \defeq \sum_{i \in [d]} \nnz(P_i)$ to refer to the total number of nonzero entries amongst all packing matrices and similarly define $\nnz(\allC)$.

\paragraph{Approximation factor.}
Throughout this work, we  explicitly assume that $\eps^{-1} \leq (dn)^{3}$, else an interior point method achieves our stated runtime.

\subsection{Useful facts}
\label{ssec:facts}

We state here facts from linear algebra and matrix calculus that we extensively use in the paper. Given a symmetric matrix $M$ with eigendecomposition $M = Q \Lambda Q^\top$ for orthonormal $Q$ and diagonal $\Lambda$, we define the matrix exponential in the standard way, $\exp(M) \defeq Q \exp(\Lambda) Q^\top$, where the exponential of a diagonal matrix is entrywise.

\paragraph{Matrix facts.} 
We use the following facts. 
\begin{itemize}
	\item For symmetric $A, B \in \Sym^n$ with $A \preceq B$, we have $\Tr\exp(A) \leq \Tr\exp(B)$.
	\item If $A, B \in \PSDSet^n$, then $\inprod{A}{B} \geq 0$. Therefore, for $M, N \in \Sym^n$ with $M \preceq N$, and $B \in \PSDSet^n$, we have $\inprod{M}{B} \leq \inprod{N}{B}$.
	\item For $A, B \in \Sym^n$ with $A \preceq B$, and arbitrary $m \times n$ matrix $R$, we have $RAR^\top \preceq RBR^\top$. This implies $A_{SS} \preceq B_{SS}$ for restrictions to a subset $S \subseteq [d]$, whether the resulting matrix is viewed as an $n \times n$ or $|S| \times |S|$ matrix.
	\item Any Schur complement of $M \in \PSDSet^n$ is also positive semidefinite. More precisely, consider a matrix $M$, which for subsets $S, L$ that partition $[n]$, can be decomposed into $M_{SS}, M_{SL}, M_{LS},$ and $M_{LL}$. Then when $M_{LL}$ is invertible, $M_{SS} - M_{SL} M_{LL}^{-1} M_{LS} \succeq 0$.
\end{itemize}

\paragraph{Matrix inequalities.}
The following matrix inequality is classical and due to \cite{Golden65, Thompson65}.

\begin{fact}[Golden-Thompson]
For $A, B \in \Sym^n$, we have $\Tr \exp(A + B) \leq \Tr[\exp(A) \exp(B)]$.
\end{fact}

The following matrix inequality inspires our potential analysis and was shown in \cite{Allen-ZhuLO16}. It generalizes the classical Lieb-Thirring inequality \cite{LiebT76}.

\begin{fact}[Extended Lieb-Thirring]
\label{fact:extendliebthirring}
For $A \succ 0, B \succeq 0$ and $\alpha \in [0, 1]$, we have
\[\inprod{B^{1/2} A^{\alpha} B^{1/2}}{B^{1/2} A^{1 - \alpha} B^{1/2}} \leq \inprod{B^2}{A}.\]
\end{fact}

We give an alternative proof of this inequality in Appendix~\ref{app:technical}, using similar ideas as a proof technique developed in this work (e.g. in showing Lemma~\ref{lem:matrixcs}).

\paragraph{Calculus.}

The following chain rule formula is due to \cite{Wilcox67}.

\begin{fact}
\label{fact:expderiv}
For a symmetric matrix-valued function $X(t)$ with argument scalar $t$, the derivative of the matrix exponential is given by
\[\frac{d}{dt}\exp(X(t)) = \int_{\alpha = 0}^1 \exp(\alpha X(t)) \frac{d}{dt} X(t) \exp((1-\alpha) X(t)) d\alpha.\] 
\end{fact}

This implies that for symmetric $A, B$, we have $\frac{d}{dt}\Tr[\exp(A + tB)] = \inprod{B}{\exp(A + tB)}$. It is well-known that $\frac{d}{dX} \Tr\exp(X) = \exp(X)$, when $\Tr\exp$ is a function on symmetric matrices. 

Finally, we state below an additive form of Gronwall's lemma that will be useful in our analysis, deferring its proof to Appendix~\ref{app:technical}.

\begin{restatable}{lemma}{restateAddGronwall}
\label{lem:addgronwall}
Let $u(t)$ be a continuously differentiable scalar-valued function defined on $t \geq 0$ which satisfies the differential inequality
\[u'(t) \geq -\beta u(t) - c \]
for some constant $c$. Then, for all $t \geq 0$,
\[u(t) \geq \exp(-\beta t) u(0) - \frac{c}{\beta}\left(1 - \exp(-\beta t)\right).\]
In particular, when $u$ is nonnegative and $|\beta t| < 1$, the following inequality holds.
\[u(t) \geq (1 - \beta t)u(0) - ct.\]
\end{restatable}

\subsection{Setup}
\label{ssec:setup}

\paragraph{Mixed positive SDP instance.}
Here, we restate the mixed positive SDP instance considered in the main body of the paper (for its relationship to the optimization problem \eqref{eq:possdpdef}, see Section~\ref{ssec:width} and Appendix~\ref{app:reduction}). For the decision variant of a mixed positive SDP, the general formulation for the guarantee of an approximation algorithm is as follows. Given inputs
\[\left\{P_i \in \PSDSet^{n_p}\right\}_{i \in [d]}, \left\{C_i \in \PSDSet^{n_c}\right\}_{i \in [d]},\]
find an $x \geq 0$ such that
\begin{equation}\lmax{\sum_{i \in [d]} x_i P_i} \leq (1 + \epsilon) \lmin{\sum_{i \in [d]} x_i C_i},\end{equation}
or demonstrate infeasibility of the constraints
\begin{equation}\sum_{i \in [d]} x_i C_i \succeq I_{n_c}, \; \sum_{i \in [d]} x_i P_i \preceq (1 - \epsilon) I_{n_p}, \; x \geq 0.\end{equation}

\paragraph{Specialization to mixed positive LP instances.}
We state the canonical positive LP decision problem considered in Section~\ref{sec:lpwarmup}. For given nonnegative matrices $P \in \R_{\geq 0}^{n_p \times d}$, $C \in \R_{\geq 0}^{n_c \times d}$, the canonical optimization problem \eqref{eq:mixedlp} reduces to the decision problem: give an $x \geq 0$ such that
\begin{equation}
\max_{j \in [n_p]} [Px]_j \leq  (1 + \epsilon) \min_{j \in [n_c]} [Cx]_j,
\end{equation}
or demonstrate infeasibility of the constraints
\begin{equation}
Cx \geq \1_{n_c},\; Px \leq \1_{n_p},\; x \geq 0.
\end{equation}
\begin{algorithm}[ht!]\caption{\textsf{MixedPositiveSDP}$(\allP,\allC,\epsilon)$}
\label{alg:sdpalg}
	\textbf{Input:} $\allP = (P_1, ..., P_d)$ and $\allC = (C_1, ..., C_d)$ where all $P_i \in \PSDSet^{n_p}$, all $C_i \in \PSDSet^{n_c}$, $n \defeq \max(n_p, n_c)$, and $\epsilon \in [(nd)^{-3},1/20]$. \\
	\textbf{Output:} ``Infeasible'' or $x \geq 0$ such that $\Px \preceq (1+\epsilon) \Cx$.
	\begin{algorithmic}[1]
		\State $K \gets \frac{4 \log(nd \rho)}{\epsilon}, \alpha \gets \frac{1}{8K}$
		\State $x_i^{(0)} \gets \frac{1}{d \lmax{P_i}}$, $\forall i \in [d]$
		\State $t \gets 0$
		\While{$\lmax{\Pxcurr} \leq K$, $\lmin{\Cxcurr} \leq K$}
		\State $Y^{(t)} \gets \exp\left(\Pxcurr\right)$
		\State $Z^{(t)} \gets \exp\left(-\Cxcurr\right)$
		\State $\pGradi \gets \inprod{P_i}{\frac{Y^{(t)}}{\Tr Y^{(t)}}}$, $\forall i \in [d]$
		\State $\cGradi \gets \inprod{C_i}{\frac{Z^{(t)}}{\Tr Z^{(t)}}}$, $\forall i \in [d]$
		\State $W^{(t)} \gets \left\{i : \pGradi \leq \left(1 - \frac{\eps}{2}\right) \cGradi \right\}$ 
		\If {$W^{(t)} = \emptyset$}
		\Return ``Infeasible''
		\EndIf
		\State Set
		\begin{equation*}
		\delta^{(t)}_i = \begin{cases}
		\frac{1}{2}\left(1 - \frac{\pGradi}{\cGradi}\right) &\text{if $i \in W^{(t)}$}\\
		0 & \text{if $i \notin W^{(t)}$}
		\end{cases}
		\end{equation*}
		\State $\xnext_i \gets \xcurr_i \left(1 + \alpha \delta^{(t)}_i\right)$, $\forall i \in [d]$
		\State $t \gets t+1$
		\EndWhile
		\State \textbf{return} ``Feasible'', $x = \xcurr$
	\end{algorithmic}
\end{algorithm}

\paragraph{Our algorithm.}

We now state the algorithm discussed in this paper for deciding mixed positive SDP instances \eqref{eq:feasible}, \eqref{eq:infeasible}. Our algorithm is exactly the same as the algorithm of \cite{MahoneyRWZ16} specialized to matrices, with the simplification that we do not explicitly prune satisfied covering constraints. In Section~\ref{sec:lpwarmup}, when we refer to its specialization to LPs, we will introduce simplified notation for working with vectors rather than matrices. 

\paragraph{Our potential.}

The following key functions appear in prior work, taking argument $M \in \Sym^n$:
\begin{equation*}
\smax(M) \defeq \lte{M}, \;\smin(M) \defeq -\lte{-M}.
\end{equation*}
The following fact is well-known and is derived from the bounded range of matrix entropy and the fact that $\smax$ and $\smin$ are appropriate conjugates:
\begin{fact}[Approximation quality of $\smax$, $\smin$]\label{fact:softmaxmin}
For all $M \in S^n$,
\[\lmax{M} \leq \smax(M) \leq \lmax{M} + \log n,\; \lmin{M} \geq \smin(M) \geq \lmin{M} - \log n.\]
\end{fact}
In Section~\ref{sec:sdpfull}, we will analyze how the following potential function changes for a particular iteration:
\begin{equation}
\label{eq:sdppotential}
\begin{aligned}
f(x^{(t)}) &\defeq (1 - \eps)\lte{\Pxcurr} + \lte{-\Cxcurr}\\
&= (1-\eps) \, \smax\left(\Pxcurr\right) - \smin\left(\Cxcurr\right).
\end{aligned}
\end{equation}
The motivation for using this potential is exactly the same as its use in earlier works \cite{Young01, MahoneyRWZ16}: it is a close approximation, up to scaling, of the key quantity 
$$\lmax{\sum_{i \in [d]} x_i P_i} - (1 + \epsilon) \lmin{\sum_{i \in [d]} x_i C_i}.$$
The main claim of Section~\ref{sec:sdpfull} is that this potential does not increase by more than a negligible amount (say, inverse-polynomial in $nd$) each iteration; we begin by demonstrating in Section~\ref{sec:lpwarmup} how to make straightforward modifications to the existing analysis to show this in the positive LP setting, and remark on the difficulties in generalizing these modifications. This will motivate and contextualize our strategy in Section~\ref{sec:sdpfull}.

With this fact about the change in the potential function in tow, we defer discussing all other aspects of the correctness and runtime of the algorithm to Section~\ref{sec:cleanup}. The remaining discussion of correctness is essentially small modifications to the correctness arguments in \cite{MahoneyRWZ16}, and the implementation details and runtime are the same as prior works on approximate semidefinite programming \cite{AroraK07, PengTZ16, Allen-ZhuLO16}. Thus, the key technical challenge is studying how this potential function changes, and comprises the bulk of this paper. 	%
\section{Linear Program Potential Analysis}
\label{sec:lpwarmup}
In this section, we give a simple analysis of the potential \eqref{eq:sdppotential} when Algorithm~\ref{alg:sdpalg} is specialized to the case of mixed positive linear programming. For just this section, we assume for simplicity of exposition that $n_p = n_c = n$; we will drop this assumption in Section~\ref{sec:sdpfull}, as it is unnecessary. We recall that in the case of LPs, Algorithm~\ref{alg:sdpalg} is effectively the same as that of \cite{MahoneyRWZ16}, but without explicit pruning of covering constraints. The potential function considered in this section has the simplified form
\begin{equation}
\label{eq:lppotential}
\begin{aligned}
f(x^{(t)}) &\defeq \log \sum_{j \in [n]} \exp\left(\left[P\xcurr\right]_j\right) + \log \sum_{j \in [n]} \exp\left(-\left[C\xcurr\right]_j\right)\\
&\defeq \, \smax\left(P\xcurr\right) - \smin\left(C\xcurr\right).
\end{aligned}
\end{equation}
Here, we overload $\smax$ and $\smin$ to their restrictions to diagonal matrices on vector inputs (i.e. we can equivalently think of a vector input as a diagonal matrix). In Lemma~\ref{lem-LPpack}, we recall the structure of the general second-order Taylor expansion argument for bounding changes to the packing potential. We then develop tools needed to give an analysis of the changing covering potential in Lemma~\ref{lem-LPcov}. Finally, we put these two bounds together to prove Theorem~\ref{thm:mainclaimlp} and conclude with some initial attempts to generalize these arguments to the positive SDP setting. The main export of this section is the following:

\begin{theorem}[Potential increase bound]
\label{thm:mainclaimlp}
On any iteration $t$ of Algorithm~\ref{alg:sdpalg}, specified to mixed positive LP instances~\eqref{eq:mixedlp}, the potential function $f$ defined in ~\eqref{eq:lppotential} satisfies
\[f(\xnext) \le f(\xcurr) + (nd\rho)^{-15}.\]
\end{theorem}

\paragraph{Notation.} We establish the notation used in this section, which we simplify because we are considering only LP instances. We consider a single iteration $t$ of the algorithm. For the context of this section, define the following: 
\begin{itemize}
	\item $x \defeq \xcurr$ to be the current iterate
	\item $\alpha\delta \circ x \defeq \xnext - \xcurr$ to be the update to the iterate
	\item $\phi \defeq Px$, $\psi \defeq Cx$
	\item $a \defeq P(\alpha\delta \circ x)$, $g \defeq P(\alpha \delta^2 \circ x)$
	\item $b \defeq C(\alpha\delta \circ x)$, $h \defeq C(\alpha\delta^2 \circ x)$
\end{itemize}
Henceforth, we drop all superscripts $(t)$ and use this notation to discuss a particular iteration. When $\exp(v)$ takes a vector-valued argument $v$ in this section, we mean the entrywise exponential.

Finally, for simplicity in this section, assume that all entries of $b$ are bounded above by $(nd\rho)^5$, and bounded below by $(nd\rho)^{-5}$. We will show how to lift this assumption in general in the analysis of Section~\ref{sec:sdpfull}, and also remove the dependence on $\rho$, in the case of positive LP instances.

\subsection{Bounding potential changes}
\label{ssec:potlp}
We first use a second-order Taylor expansion and Cauchy-Schwarz to control the change in the packing potential, $\smax(Px) = \smax(\phi)$, from $Px = \phi$ to $P((\1_d + \alpha\delta)\circ x) = \phi + a$.

\begin{lemma}[Packing potential]\label{lem-LPpack} 
We have the following bound on the packing potential change:
$$\log \inprod{\1_n}{\exp(\phi + a)} \leq \log \inprod{\1_n}{\exp(\phi)} + \inprod{a + g}{\frac{\exp(\phi)}{\1_n^\top\exp(\phi)}} .$$
\end{lemma}
\begin{proof}
First, by the termination condition, $\phi \leq K$ entrywise. Further, because $0 \leq \alpha\delta \leq \frac{1}{16K}$ entrywise, $a$ is bounded entrywise by $\frac{1}{16}$ as it is a scaling of $\phi$ by $\alpha\delta$. Therefore, by the bound $e^x \leq 1 + x + x^2$ for $0 < x < 1$ for all $j \in [n]$,
\begin{equation}\label{eq:expandlppack}\exp(\phi_j + a_j) = \exp(\phi_j) \exp(a_j) \leq \exp(\phi_j) (1 + a_j + a_j^2) \leq \exp(\phi_j)(1 + a_j + g_j).\end{equation}
In the last inequality, we also noted that by the definitions of $a$ and $g$ and Cauchy-Schwarz, 
\begin{align}
\label{eq:cauchyschwarzsmall}
a_j^2 = \alpha^2 \left(\sum_{i \in [d]} P_{ji} x_i \delta_i\right)^2 \leq \left(\alpha\sum_{i \in [d]} P_{ji} x_i \delta_i^2\right) \left(\alpha\sum_{i \in [d]} P_{ji} x_i\right) \leq \alpha K g_j \leq g_j.
\end{align}
Summing~\eqref{eq:expandlppack} over all $j$ yields
\[\inprod{\1_n}{\exp(\phi + a)} \leq \inprod{\1_n}{\exp(\phi)} \left(1 + \inprod{a + g}{\frac{\exp(\phi)}{\1_n^\top\exp(\phi)}} \right);\]
taking a logarithm and using $\log(1 + x) \leq x$ yields the desired claim.
\end{proof}

Attempts to follow this line of reasoning to obtain a lower bound on the change in the covering potential run into the immediate issue that the algorithm guarantees only $\min_j \psi_j \leq K$, where we recall $\psi = Cx$, which means the Taylor expansion used in \eqref{eq:expandlppack} does not hold for \textit{all} coordinates. This was the reason \cite{Young01, MahoneyRWZ16} explicitly pruned constraints. We sidestep this issue by noting that any sufficiently large coordinate of $\psi$ will necessarily make an extremely small contribution to $\smin(Cx) = -\log \inprod{\1_n}{\exp(-\psi)}$, quantifying the claim as follows. 

\begin{lemma}\label{lem-totsum-goodsum} For constants $c$, $T \ge 0$, vector $y \in \R^n$ satisfying $\min_j y_j < T$, and vector $v\in \R^n$ entrywise satisfying $\ell \leq v \leq u$ for scalars $0 < \ell < u$, define the following two sets:
	\[S \defeq \{j \in [n] \mid y_j \leq T + c\}, \; L \defeq \{j \in [n] \mid y_j \geq T + c\}.\]
Then, we have
	\[ \sum_{j\in [n]} \exp(-y_j) v_j \leq \left( 1 + \frac{n u }{\ell} \exp(-c)\right)\sum_{j \in S} \exp(-y_j) v_j.\] 
\end{lemma} 
\begin{proof}
	Choose some coordinate $j^*$ such that $y_{j^*} < T$.  By definition, for every $j \in L$, we have 
	$$y_j > y_{j^*} + c \implies \exp(-y_j) < \exp(-y_{j^*}) \exp(-c).$$ 
	We therefore have the following sequence of inequalities:
	\begin{align*}
	\sum_{j \in L} \exp(-y_j) v_j &\leq \sum_{j \in L} \exp(-y_{j^*}) \exp(-c) v_j \\
	&\leq \frac{nu}{\ell} \exp(-c) \exp(-y_{j^*})  v_{j^*} \\
	&\leq \frac{nu}{\ell} \exp(-c) \sum_{j \in S} \exp(-y_j) v_j.
	\end{align*} 
	The claim follows from $\sum_{j \in [n]} \exp(-y_j) v_j = \sum_{j \in S }\exp(-y_j) v_j  + \sum_{j \in L} \exp(-y_j) v_j$.
\end{proof}

Lemma~\ref{lem-totsum-goodsum} gives us a way of bounding the relative contributions of the \emph{large} coordinates of a vector $y$ to a weighted sum of inverse exponentials of \emph{all} its coordinates. In particular, for say $c = O(\log n)$ sufficiently large and $u/\ell$ bounded by a polynomial in $n$, the lemma states that the relative contribution of the large coordinates is negligible. 

We give a concrete application in the following bound on the covering potential: intuitively, we do a Taylor expansion on coordinates we can control by a magnitude bound, and argue that adding in the remaining coordinates does not significantly affect the value of the potential function.

\begin{lemma}
\label{lem-LPcov}
We have the following bound on the covering potential change:
\[-\log \inprod{\1_n}{\exp(-\psi - b)} \geq -\log \inprod{\1_n}{\exp(-\psi)} + \inprod{b - h}{\frac{\exp(-\psi)}{\1_n^\top\exp(-\psi)}}  - (nd\rho)^{-15}.\]
\end{lemma}
\begin{proof}
Define the sets
\[S \defeq \{j \in [n] \mid \psi_j \leq 2K\}, \; L \defeq \{j \in [n] \mid \psi_j \geq 2K\}.\]
For a vector $v \in \R^n$, we denote $v_S$ to be the $n$-dimensional vector whose coordinates $j \not\in S$ are zeroed, and similarly define $v_L$. Further, we note that every entry of $b_S$ is bounded by $\frac{1}{4}$, since it is a scaling of $\psi_S$ with scaling factor $\alpha\delta$ at most $\frac{1}{16K}$ coordinatewise. For each $j \in S$, we have by the bound $e^{-x} \leq 1 - x + x^2$ for $0 < x < 1$,
\begin{equation}
\label{eq:expandlpcov}
\exp(-\psi_j - b_j) = \exp(-\psi_j) \exp(-b_j) \leq \exp(-\psi_j)(1 - b_j + b_j^2) \leq \exp(-\psi_j)(1 - b_j + h_j).
\end{equation}
Here, we used a similar Cauchy-Schwarz bound as in \eqref{eq:cauchyschwarzsmall} to show $b_j^2 \leq h_j$ for $b_j \in S$. Next, using Lemma~\ref{lem-totsum-goodsum} with $v = \1_n$, $u = l = 1$, $T = 1.5K$, and $c = 0.5K$ yields
\[\inprod{\1_n}{\exp(-\psi - b)} \leq (1 + (nd\rho)^{-31})\inprod{\1_n}{\exp([-\psi  - b]_S)}. \]
To see that this was a valid usage, recall our definition of $L$, and note that $\min_j \psi_j \leq K$, and $\psi + b$ is multiplicatively bounded by $(1 + \frac{1}{16K})\psi \leq 1.5\psi$. We also use the loose bound $n\exp(-c) \leq n(nd\rho)^{-2/\eps} \leq (nd\rho)^{-31}$. Now, summing \eqref{eq:expandlpcov} over all $j \in S$ yields
\begin{equation}
\label{eq:halfwaylpcov}
\begin{aligned}
\inprod{\1_n}{\exp(-\psi - b)} &\leq (1 + (nd\rho)^{-31})\inprod{\1_n}{\exp([-\psi  - b]_S)} \\
&\leq (1 + (nd\rho)^{-31})\inprod{\1_n - b + h}{[\exp(-\psi)]_S}\\
&\leq (1 + (nd\rho)^{-31}) \inprod{\1_n}{\exp(-\psi)}  \left( 1- \inprod{b}{\frac{[\exp(-\psi)]_S}{\1_n^\top \exp(-\psi)}} + \inprod{h}{\frac{[\exp(-\psi)]_S}{\1_n^\top \exp(-\psi)}} \right).
\end{aligned}
\end{equation}

Taking negative logarithms gives the following bound:
\begin{align*}
-\log \inprod{\1_n}{\exp(-\psi - b)} &\geq -\log \inprod{\1_n}{\exp(-\psi)} + \inprod{b}{\frac{[\exp(-\psi)]_S}{\1_n^\top \exp(-\psi)}} - \inprod{h}{\frac{[\exp(-\psi)]_S}{\1_n^\top \exp(-\psi)}} - (nd\rho)^{-30} \\
&\geq -\log \inprod{\1_n}{\exp(-\psi)} + \inprod{b}{\frac{[\exp(-\psi)]_S}{\1_n^\top \exp(-\psi)}} - \inprod{h}{\frac{\exp(-\psi)}{\1_n^\top \exp(-\psi)}} - (nd\rho)^{-30}.
\end{align*}
Here, we used the facts $-\log(1- x) > x$ for $0 < x < 1$, and $-\log(1 + x) > -2x$ for $x > 0$. Finally, to bound the middle term, we use Lemma~\ref{lem-totsum-goodsum} one more time: by our assumptions on $b$ and the uniform scaling $\1_n^\top \exp(-\psi)$, we may set $v = b$, $u/\ell \leq (nd\rho)^{10}$, $T = K$, and $c = K$, yielding
\begin{equation*}
\begin{aligned}
\inprod{b}{\frac{[\exp(-\psi)]_S}{\1_n^\top \exp(-\psi)}} \geq \left(1 - (nd\rho)^{-25}\right)\inprod{b}{\frac{\exp(-\psi)}{\1_n^\top \exp(-\psi)}}.
\end{aligned}
\end{equation*}
Here, we used Lemma~\ref{lem-totsum-goodsum} with  the loose bound $\frac{nu}{\ell} \cdot \exp(-c) \leq (nd\rho)^{11} (nd\rho)^{-40} \leq (nd\rho)^{-25}$. Finally, since $\exp(-\psi)/\1_n^\top \exp(-\psi)$ is a probability distribution,
\begin{equation*}
-(nd\rho)^{-25} \inprod{b}{\frac{\exp(-\psi)}{\1_n^\top \exp(-\psi)}} \geq -(nd\rho)^{-20},
\end{equation*}
where we recalled the assumption that the largest entry of $b$ is bounded by $(nd\rho)^5$. This concludes the proof, by $(nd\rho)^{-20} + (nd\rho)^{-30} \leq (nd\rho)^{-15}$.
\end{proof}

We now combine Lemma~\ref{lem-LPpack} and Lemma~\ref{lem-LPcov} to prove Theorem~\ref{thm:mainclaimlp}. Here, we recall that by our notation, $Px^{(t)} = \phi$, $Cx^{(t)} = \psi$, $Px^{(t + 1)} = \phi + a$, and $Cx^{(t + 1)} = \psi + b$.

\begin{proof}[Proof of Theorem~\ref{thm:mainclaimlp}]
By combining Lemma~\ref{lem-LPpack} and Lemma~\ref{lem-LPcov}, and using the definitions 
$$f(x) = \log \inprod{\1_n}{\exp(\phi)} + \log \inprod{\1_n}{\exp(-\psi)}$$
and
$$f(x \circ (1 + \alpha\delta)) = \log \inprod{\1_n}{\exp(\phi + a)} + \log \inprod{\1_n}{\exp(-\psi - b)},$$
the potential increase is at most $(nd\rho)^{-15}$ as long as we can show that
	\begin{align} 
	\label{eq:nonnegativelp}
	\inprod{a + g}{ \frac{\exp(\phi)}{\1_n^\top \exp(\phi)}} - \inprod{b - h}{\frac{\exp(-\psi)}{\1_n^\top \exp(-\psi)}} \leq 0.
	\end{align}
	To do so, we note the form of the gradient updates,
	$$\pgradplaini = \inprod{\pii}{\frac{\exp(\phi)}{\1^\top \exp(\phi)}},\; \cgradplaini = \inprod{\ci}{\frac{\exp(-\psi)}{\1^\top \exp(-\psi)}}.$$
	Thus, by the definitions of $a, g, b, h$, we rewrite the left hand side of \eqref{eq:nonnegativelp} as
	\begin{align}
	\label{eq:nonnegativelp2}
	\alpha \sum_{i\in[d]} \delta_i x_i \left((1+ \delta_i)\pgradplaini - (1-\delta_i) \cgradplaini\right).
	\end{align}
	By its definition in Algorithm~\ref{alg:sdpalg}, either $\delta_i = 0 $ or $\delta_i= \frac{1}{2}(1 - \pgradplaini/\cgradplaini)$. For each summand $i$ in \eqref{eq:nonnegativelp2}, if $\delta_i = 0$, it is clearly nonpositive; otherwise, 
	\[(1+ \delta_i)\pgradplaini - (1-\delta_i) \cgradplaini = \frac{3 \pgradplaini \cgradplaini - (\pgradplaini)^2}{2\cgradplaini} - \frac{\pgradplaini + \cgradplaini}{2} = \frac{(\pgradplaini - \cgradplaini)^2}{2\cgradplaini} \leq 0.\] 
	By nonnegativity of $\delta_i$ and $x_i$, summing over all $i \in [d]$ yields the desired claim.
\end{proof}
\subsection{Generalizing to a matrix potential analysis}\label{ssec:generalizingIdeas}
In this section, as a precursor to Section~\ref{sec:sdpfull} and to motivate our subsequent developments, we describe some first-pass attempts at generalizing the proofs in Section~\ref{ssec:potlp} to analyze the SDP potential. By a matrix generalization of Cauchy-Schwarz that we develop, the proof of Lemma \ref{lem-LPpack} can be extended to the SDP setting in a relatively straightforward manner; however, due to monotonicity conditions that do not generalize from the vector case to non-commuting matrices, it is less clear how to adapt the proof of Lemma \ref{lem-LPcov}. 

\paragraph{Attempt 1: Projecting onto the small eigenspace.} Recall that in the proof of Lemma~\ref{lem-LPcov}, a key idea for dealing with the coordinates that did not admit a Taylor expansion was to first perform a projection onto the small set of coordinates. To this end, a natural idea is to let $P_{\Psi}$ be the projection matrix corresponding to the span of the eigenvectors of $\Psi$ whose eigenvalues are smaller than a threshold, e.g. $2K$. Then, we can argue directly that
$$P_{\Psi} B P_{\Psi} \preceq \frac{\alpha}{2} P_{\Psi} \Psi P_{\Psi} \preceq \frac{1}{4} I_n.$$
This immediately motivates using an argument of the form
\begin{align*}
\Tr \exp(-\Psi - B) &\leq \Tr\left[\exp(-\Psi)\exp(-B)\right]\\
&\leq \Tr\left[\exp(-\Psi) \exp(-P_{\Psi} B P_{\Psi})\right]\\
&\leq \Tr\left[\exp(-\Psi)\left(I_n -P_{\Psi} B P_{\Psi} + (P_{\Psi} B P_{\Psi})^2 \right)\right].
\end{align*} 
This sequence of inequalities is analogous to \eqref{eq:expandlpcov}; the first inequality is by Golden-Thompson, and the third is by the spectral norm bound on $P_{\Psi} B P_{\Psi}$. However, the second inequality is incorrect; in general, it is not true that for a symmetric projection matrix $P_{\Psi}$, $P_{\Psi} B P_{\Psi} \preceq B$ when the matrices do not commute. Immediate attempts to design more carefully-defined projection matrices to use the sort of ``small bucket'' arguments as in the proof of Lemma~\ref{lem-LPcov} similarly fail due to monotonicity reasons, motivating the search for a ``commutative'' version of this strategy. \medskip

\paragraph{Attempt 2: Local projections via known trace inequalities.} The strategy of \cite{Allen-ZhuLO16} was to define an interpolation function, relating a known bound to a desired bound, and locally bounding the error. More formally, it can be shown that the vector function $g(t) = \Tr\exp(-\Psi - tB)$ is convex in $t$. Therefore, it satisfies $g(1) - g'(1) \leq g(0)$, or equivalently
\begin{align*}
\Tr\exp(-\Psi - B) \leq \Tr \exp(-\Psi) - \inprod{B}{\exp(-\Psi - B)}.
\end{align*} 
Therefore, if we were able to show that at least the first-order relationship
 \[\inprod{B}{\exp(-\Psi - B)} \geq (1-\Oh{\epsilon})\inprod{B}{\exp(-\Psi)}\] holds, it is not hard to see that we could scale down the step size by a factor $\eps$ and give a $\tOh{\eps^{-4}}$ algorithm. The strategy used in \cite{Allen-ZhuLO16} was to use Gr\"onwall's lemma, e.g. Lemma~\ref{lem:addgronwall}, to show a bound of the form
 \[-g'(t) \le \eps g(t).\] 
Because this inequality is local to a particular $t \in [0, 1]$, it is simpler to relate the quantities on the two sides. By applying Fact~\ref{fact:expderiv}, we have
\begin{align*}
\frac{d}{dt} g(t) &= -\inprod{B}{\int_0^{1} e^{s(-\Psi - tB)} B e^{(1 - s)(-\Psi - tB) }ds} \\
&\geq -\Tr(B^2 \exp(-\Psi - tB)).
\end{align*} 
Here, for each $s \in [0, 1]$, we applied the extended Lieb-Thirring inequality (Fact~\ref{fact:extendliebthirring}). Thus, it now suffices to show a relationship such as
\begin{equation}\label{eq:falsesecondorder}\Tr\left[B^2\exp(-\Psi - tB)\right] \leq \eps \Tr\left[B\exp(-\Psi - tB)\right].\end{equation}
Here, one may hope to apply an argument such as Lemma~\ref{lem-totsum-goodsum} after diagonalizing the matrices $\Psi + tB$ and $B$, to quantify if the error incurred by the large eigenvalues of $B$ is dominated by their exponential weighting in $\exp(-\Psi - tB)$. However, it is easy to check that \eqref{eq:falsesecondorder} is false; matrices
\[\Psi = \begin{pmatrix}1 & 0 \\ 0 & M\end{pmatrix} \text{ and } B = \frac{\eps}{2}\begin{pmatrix}1 & \sqrt{M} \\ \sqrt{M} & M\end{pmatrix}\]
break this inequality even for $t = 0$, for large choices of $M$, despite $B \preceq \eps \Psi$ holding.
\medskip

Ultimately, our analysis in Section~\ref{sec:sdpfull} is based on a careful combination of these ideas: we design a second-order interpolation function and remove the black-box use of the extended Lieb-Thirring inequality in favor of a more fine-grained inequality, similar in flavor to Lemma~\ref{lem-totsum-goodsum}. 	%
\section{Semidefinite Program Potential Analysis}
\label{sec:sdpfull}

As mentioned in Section \ref{sec:lpwarmup}, the two main technical challenges when extending \cite{MahoneyRWZ16} to the SDP setting are correctly arguing a spectral bound on the step $B = \sum_i x_i C_i$ restricted to a ``small'' subspace, and going beyond a black-box application of the extended Lieb-Thirring inequality. We discuss our techniques for overcoming these difficulties, as well as obtaining a second-order bound crucial for matching the iteration count of \cite{MahoneyRWZ16}. We rely on an extension of the Cauchy-Schwarz inequality to positive semidefinite matrices. This inequality is a specialization of Kadison's inequality on unital positive maps, and its proof is representative of the matrix inequality techniques which appear in this section.

\begin{lemma}[Matrix Cauchy-Schwarz]
\label{lem:matrixcs}
Let $M_i \in \mathbb{S}^n_{\ge 0}$ be such that $\sum_i M_i \preceq KI_n$ for some scalar $K$. Let $c_i \geq 0$ be scalars. Then,
$$\left(\sum_i c_i M_i\right)^2 \preceq K \sum_i c_i^2 M_i.$$
\end{lemma}

\begin{proof}
First, note that for any vectors $u, v$, and $\tilde{u} = M_i^{1/2} u, \tilde{v} = M_i^{1/2} v$,
\begin{align*}
\begin{pmatrix} u^\top & v^\top \end{pmatrix} \begin{pmatrix} c_i^2 M_i & c_i M_i \\ c_i M_i & M_i\end{pmatrix} \begin{pmatrix} u \\ v \end{pmatrix} = \norm{c_i \tilde{u} + \tilde{v}}_2^2 \geq 0 \\
\Rightarrow \begin{pmatrix} c_i^2 M_i & c_i M_i \\ c_i M_i & M_i\end{pmatrix} \succeq 0 \\
\Rightarrow \begin{pmatrix} \sum_i c_i^2 M_i & \sum_i c_i M_i \\ \sum_i c_i M_i & \sum_i M_i \end{pmatrix} \succeq 0
\Rightarrow \begin{pmatrix} \sum_i c_i^2 M_i & \sum_i c_i M_i \\ \sum_i c_i M_i & KI_n \end{pmatrix} \succeq 0
\end{align*}
where in the last line we used that $\sum_i M_i \preceq K I_n$. Now, the Schur complement with respect to the $KI_n$ block is $\sum_i c_i^2 M_i - \frac{1}{K}\left(\sum_i c_i M_i\right)^2$, which is positive semidefinite since any Schur complement of any positive semidefinite matrix is positive semidefinite. Rearranging implies the result.
\end{proof}

We also require one useful technical property about our algorithm, which will factor into the analysis in this section; we defer its proof to Appendix~\ref{app:technical}.

\begin{restatable}{lemma}{restateEntryBound}
	\label{lem:entrybound}
	Throughout Algorithm~\ref{alg:sdpalg}, no entry of
	$$\Cxcurr, \; \alpha \sum_{i = 1}^{d} \delta_i^{(t)} x_i^{(t)} C_i, \; \text{ or } \alpha \sum_{i = 1}^{d} \left(\delta_i^{(t)}\right)^2 x_i^{(t)} C_i$$
	has magnitude larger than $Knd\rho \leq (nd\rho)^5$.
\end{restatable}

To overcome the obstacle of our covering potential being unbounded, we open up our potential function and separately analyze the contribution along the large and small eigendirections of our covering variable. Crucially, we use the technique of \cite{Allen-ZhuLO16} in controlling the derivative of an interpolating function, in order to locally control the contribution of the large directions of $B$. The goal of this section is to prove the following main claim, the analog to Theorem~\ref{thm:mainclaimlp}, where we recall in this section the definition of the potential function we use:
\begin{align*}f(x^{(t)}) &\defeq (1 - \eps)\lte{\Pxcurr} + \lte{-\Cxcurr}\\
&= (1-\eps) \, \smax\left(\Pxcurr\right) - \smin\left(\Cxcurr\right).\end{align*}
\begin{theorem}[Potential increase bound]
    \label{thm:mainclaimsdp}
    On any iteration $t$ of Algorithm~\ref{alg:sdpalg}, 
    \[f(\xnext) \le f(\xcurr) + (nd\rho)^{-15}.\]
\end{theorem}

We start by fixing the notation used in this section, analogous to that of Section~\ref{sec:lpwarmup}; we will consider only a particular iteration $t$, so after fixing this notation we will drop the superscript $(t)$.

\begin{definition}
For iteration $t$ of the algorithm, let $\Phi = \Pxcurr$ and $\Psi = \Cxcurr$. Let $\packstep = \sum_{i=1}^d \alpha \delta_i^{(t)} x_i^{(t)} P_i$ and $\covstep = \sum_{i=1}^d \alpha \delta_i^{(t)} x_i^{(t)} C_i$ be the change to $\Phi$ and $\Psi$ after a step on $x$. Additionally, let $\packsq = \sum_{i=1}^d \alpha (\delta_i^{(t)})^2 x_i^{(t)} P_i$ and $\covsq = \sum_{i=1}^d \alpha (\delta_i^{(t)})^2 x_i^{(t)} C_i$. 
\end{definition}

The matrices $\packsq$ and $\covsq$ are the matrix equivalents of the second-order term in the potential analysis from Section \ref{sec:lpwarmup}. We observe the following properties of the above matrices:

\begin{lemma}\label{lemma:helper}
The matrices $\Phi, \Psi, \packstep, \covstep, \packsq, \covsq$ satisfy
\begin{itemize}
\item $\packstep \preceq \frac{\alpha}{2} \Phi $
\item $\covstep \preceq \frac{\alpha}{2} \Psi$
\item $\covsq \preceq \frac{1}{2} \covstep$
\item $\packstep \preceq \frac{1}{2} \inp$
\item $\packstep^2 \preceq \packsq$
\end{itemize}
\end{lemma}

\begin{proof}
We prove the claims in order. For the first two conditions, observe that $\delta_i^{(t)}$ is at most $\frac{1}{2}$. Thus $\alpha \delta_i^{(t)} \leq \frac{\alpha}{2}$: the claim follows from the definitions.

For the third condition, we again have $\delta_i^{(t)} \leq \frac{1}{2}$, and the claim again follows from the definitions. 

For the fourth condition, we observe that the algorithm's termination condition ensures $\Phi \preceq K I_{n_p}$. Thus $\packstep \preceq \frac{\alpha K}{2} \inp \preceq \frac{1}{2} \inp$. 
    
For the fifth condition, we apply Lemma \ref{lem:matrixcs} with $c_i \defeq \alpha \delta_i^{(t)}$ and $M_i \defeq x_i^{(t)} P_i$. We have $\sum_{i=1}^d M_i = \Phi \preceq K \inp$, and so we have 
\[
\packstep^2 = \left( \sum_{i=1}^d c_i M_i \right)^2  \preceq K \left( \sum_{i=1}^d c_i^2 M_i \right) \preceq \packsq,
\]
where in the last statement we use that $K c_i^2  = K \alpha^2 (\delta_i^{(t)})^2 \leq  \alpha (\delta_i^{(t)})^2$. 

\end{proof}

For Sections \ref{sc:packing} and \ref{sc:covering}, we will fix the iteration $t$ and analyze the change in our potential functions for this single iteration. At the end of this section we will combine our analyses and prove our main potential function bound. 

\subsection{Packing potential bound}
\label{sc:packing}
We begin by proving the packing potential bound, as this gives a good primer for the complications that the covering potential will later induce. In particular, we show the following claim:

\begin{lemma}[Packing potential bound] \label{lemma:packing}
For any fixed iteration, 
\[
\log\te{\Phi + \packstep} \leq \log\te{\Phi} + \trprod{\packstep + \packsq}{\frac{\exp(\Phi)}{\te{\Phi}}}.
\]
\end{lemma}
\begin{proof}
We have
\begin{align*}
\te{\Phi + \packstep} &\leq  \trprod{\exp(\Phi)}{\exp(\packstep)}\\
&\leq  \trprod{\exp(\Phi)}{\inp + \packstep + \packstep^2}\\
&\leq \trprod{\exp(\Phi)}{\inp + \packstep + \packsq} \\
&= \te{\Phi}\left(1 + \inprod{A + G}{\frac{\exp(\Phi)}{\te{\Phi}}} \right).
\end{align*}
Here, the first statement follows by Golden-Thompson. The second statement holds via monotonicity of the trace product, the fact that $e^x \leq 1+ x+x^2$ for all $0 < x < 1$, and $\packstep \preceq \frac{1}{2} \inp$ (Lemma \ref{lemma:helper} part 3). The final inequality follows since $\packstep^2 \preceq \packsq$ by Lemma \ref{lemma:helper} part 4. Now, taking a logarithm of both sides and using $\log(1 + x) \le x$ for $x > 0$ yields the result.
\end{proof}

\subsection{Covering potential bound}
\label{sc:covering}
In contrast to the above analysis of the packing potential, our analysis of the change in the covering potential is much more involved. The goal of this section is to prove the following bound.

\begin{restatable}[Covering potential bound]{lemma}{restateCov} \label{lemma:covering}
For any iteration of our algorithm, we have
\[
\log\te{-\Psi-\covstep} \leq \log\te{-\Psi} - (1-\epsilon)\trprod{\covstep - \covsq}{\frac{\exp(-\Psi)}{\te{-\Psi}}} + (nd\rho)^{-15}.
\]
\end{restatable}

\subsubsection{Main ideas of proof}
As mentioned before, the main difficulty in proving this bound is our lack of a suitable upper bound on the spectral magnitude of our step matrices $\sum_i \delta_i x_i C_i$, creating difficulty in directly using a Taylor approximation.  We now sketch our approach for proving this bound. Define the function $g(t) \defeq \te{-\Psi - t\covstep}$. Our goal is to compare $g(1)$ and $g(0)$. As $g$ is a convex function in $t$, we obtain $g(1) \leq g(0) + g'(1)$; our goal is to upper bound $g'(1)$. To do this, we define the interpolation function 
\[
\iota(t) = \trprod{\covstep}{\exp(-\Psi-t \covstep)} - (1-t) \trprod{\covsq}{\exp(-\Psi-t \covstep)}.
\]

This function will serve as a stand-in for the gradient terms in our bound. It is constructed to "interpolate" between the $g'(1)$ in the bound that we have from convexity and the actual expression we want in Lemma \ref{lemma:covering}. Indeed, we observe that
\[
\iota(1) = \trprod{\covstep}{\exp(-\Psi-\covstep)}, \; \iota(0) = \trprod{\covstep}{\exp(-\Psi)} - \trprod{\covsq}{\exp(-\Psi)}.
\]

Thus, our goal will be to show $(1-\epsilon)\iota(0)\leq \iota(1)$, up to small additive error. We will do this by carefully bounding the derivatives $\iota'(t)$ in terms of $\iota(t)$: this bound combined with an additive-error formulation of Gr\"onwall's lemma will allow us to prove our desired bound.

\subsubsection{Derivatives of $\iota(t)$ and bucketing}

As discussed above, our analysis requires to carefully bound the derivatives of $\iota(t)$. We begin by obtaining a closed form for them.

\begin{lemma}\label{lemma:iotaderiv-integral}
For any $0 \leq t \leq 1$, 
\begin{equation*}
\begin{aligned}
\iota'(t) &= - \trprod{\covstep}{\int_0^1 e^{-\alpha(\Psi + t\covstep)} \covstep e^{-(1-\alpha)(\Psi + t\covstep)}d\alpha}  \\
&+(1 - t) \trprod{\covsq}{\int_0^1 e^{-\alpha(\Psi + t\covstep)} \covstep e^{-(1-\alpha)(\Psi + t\covstep)} d\alpha} \\
&+ \trprod{\covsq}{\exp(-\Psi - t\covstep)}.
\end{aligned}
\end{equation*}
\end{lemma}
\begin{proof}
As $\iota(t) = \trprod{\covstep}{\exp(-\Psi-t \covstep)} - (1-t) \trprod{\covsq}{\exp(-\Psi-t \covstep)}$, we take the derivative of $\iota$ termwise. A straightforward application of Fact \ref{fact:expderiv} shows that the derivative of $\trprod{\covstep}{\exp(-\Psi-t \covstep)}$ is equal to the first integral in the statement. Similarly, a combination of the product rule and Fact \ref{fact:expderiv} shows that the derivative of $- (1-t)\trprod{\covsq}{\exp(-\Psi-t \covstep)}$ equals the last two terms. The claim follows.
\end{proof}

With this formulation, we further simplify $\iota'(t)$ and express it as a sum over $i,j \in [n_c]$. This technique has found use in the literature of positive semidefinite programming \cite{Allen-ZhuLO16}.

\begin{lemma}\label{lemma:iotaderiv}
For any $0 \leq t \leq 1$, let $Q \Lambda Q^\top$ be the eigendecomposition of $\Psi + t \covstep$ where $Q$ is orthonormal and $\Lambda$ is diagonal. Let $\rcovstep = Q^\top \covstep Q$ and $\rcovsq = Q^\top \covsq Q$. Let 
\[
    \nu(\beta, \gamma) = \begin{cases}
        \frac{e^{-\beta} - e^{-\gamma}}{\gamma - \beta} &\text{if $\gamma \neq \beta$}\\
        e^{-\gamma} & \text{if $\gamma = \beta$}
        \end{cases}
\] 
Then 

\begin{align*}
-\iota'(t) = \sum_{i, j} (\rcovstep_{ij})^2 \nu(\lambda_i,\lambda_j) - (1 - t) \sum_{i, j} \rcovstep_{ij} \rcovsq_{ij} \nu(\lambda_i,\lambda_j)  - \trprod{\rcovsq}{e^{-\Lambda}}.
\end{align*}
\end{lemma}
\begin{proof}
We start by remarking that the function $\nu(\beta,\gamma)$ is continuous in its arguments. By substituting in the eigendecomposition of $\exp(-\Psi - t\covstep)$ into Lemma \ref{lemma:iotaderiv-integral} and rearranging we obtain
\begin{align*}
-\iota'(t) = \trprod{\rcovstep}{\int_0^1 e^{-\alpha \Lambda} \rcovstep e^{(1-\alpha)\Lambda}d\alpha} - (1 - t) \trprod{\rcovsq}{\int_0^1 e^{-\alpha \Lambda} \rcovstep e^{(1-\alpha)\Lambda}d\alpha}  - \trprod{\rcovsq}{\exp(-\Lambda)}.
\end{align*}
Expanding the trace products into sums over the indices, we get 
\begin{align*}
-\iota'(t) = \sum_{i, j} (\rcovstep_{ij})^2 \int_0^1 e^{-\alpha \lambda_i - (1-\alpha) \lambda_j} d\alpha - (1 - t) \sum_{i, j} \rcovstep_{ij} \rcovsq_{ij} \int_0^1 e^{-\alpha \lambda_i - (1-\alpha) \lambda_j} d\alpha  - \trprod{\rcovsq}{e^{-\Lambda}}.
\end{align*}
As $\int_0^1 e^{-\alpha \lambda_i - (1-\alpha) \lambda_j} d\alpha = \nu(\lambda_i,\lambda_j),$ the result follows. 
\end{proof}
With this formulation of $\iota'(t)$ in hand, we describe how we bound it. Our goal is to construct an upper bound on $-\iota'(t)$ in terms of $\iota(t)$. We do this by splitting the set of pairs of eigenvalues $(i,j) \in [n_c] \times [n_c]$ we sum over into three parts.

\begin{definition}[Index buckets]
    Given a collection of eigenvalues $\lambda_1, \lambda_2, ... \lambda_{n_c}$, we define
    \begin{itemize}
        \item The small bucket $S \defeq \{(i,j) : \lambda_i, \lambda_j \leq 4K \}$
        \item The gap bucket $\Gamma \defeq \{(i,j) : \lambda_i \leq 2K, \lambda_j > 4K \text{ or } \lambda_j \leq 2K, \lambda_i > 4K\}$
        \item The large bucket $L \defeq \{(i,j) : 2K < \lambda_i, \lambda_j > 4K \text{ or } 2K < \lambda_j, \lambda_i > 4K \}$
    \end{itemize}
\end{definition}

Observe that these buckets are disjoint and cover the entire space of possible ordered pairs $(i,j) \in [n_c] \times [n_c]$. We split our summations into these three buckets and separately consider each bucket. Intuitively, the summation restricted to the large bucket must be small: by Lemma \ref{lem:entrybound} we can get a polynomial upper bound on the entries of $\rcovstep$ and $\rcovsq$, while threshold $K$ is chosen so that $e^{-2K}$ dominates this. Further, the summation over the small bucket can be interpreted as applying a projection onto the small eigendirections of $\exp(-\Psi - t\covstep)$: in this case we can essentially operate as if $\covstep$ were bounded from above just as we did for the packing potential. Finally, we bound the sum over the gap terms by bounding from above the size of the off-diagonal entries of $\rcovstep$ and $\rcovsq$ in terms of the diagonal entries. 

\subsubsection{Large bucket}
\label{ssec:largegap}
We begin by computing an upper bound for the summations restricted to the large bucket.
\begin{lemma}[Large bucket bound]
    \label{lemma:largebucket}
For any $0 \leq t \leq 1$, 
\[
    \sum_{(i,j) \in L} (\rcovstep_{ij})^2 \nu(\lambda_i,\lambda_j) \leq (nd\rho)^{-60}\Tr\exp(-\Psi)
\]
and
\[
    \left| \sum_{(i,j) \in L} \rcovstep_{ij} \rcovsq_{ij} \nu(\lambda_i,\lambda_j) \right| \leq (nd\rho)^{-60}\Tr\exp(-\Psi).
\]
\end{lemma}
\begin{proof}
By Lemma \ref{lem:entrybound}, we have $|| \covstep||_2 \leq n (nd\rho)^5$ and $||\covsq||_2  \leq n (nd\rho)^5$, since the operator norm of a matrix is at most the trace. As $\rcovstep$ and $\rcovsq$ are obtained by orthonormal rotations of $\covstep$ and $\covsq$, we see that for any indices $i,j$, $|\rcovstep_{ij}|, |\rcovsq_{ij}| \leq n (nd\rho)^5$, since again we can control the diagonal entries by the trace, and all off-diagonal entries of a PSD matrix can be bounded by the diagonal. In addition, since $\lambda_i, \lambda_j > 2K$ for any $i,j \in L$, we have by Lemma \ref{lemma:nubound} that 
\[\nu(\lambda_i,\lambda_j) \leq \frac{e^{-\lambda_i} + e^{-\lambda_j}}{2} \leq e^{-2K} \leq (nd\rho)^{-80}\Tr\exp(-\Psi).\]
We used the bounds (as $\Psi$ has at least one eigenvalue smaller than $K$ by the termination condition)
\[e^{-K} \le (nd\rho)^{-\frac{4}{\eps}} \le (nd\rho)^{-80},\; \Tr\exp(-\Psi) \ge e^{-K}.\]
As both sums in the lemma statement sum over at most $n^2$ pairs of indices, the claim follows.

\end{proof}

We remark that this bound can be strengthened to $(nd)^{-60}$ (i.e. with no dependence on the quantity $\rho$), and we can also choose threshold $K$ in the algorithm independent of $\rho$, in the case where all covering matrices commute. This is because we can directly bound this using a trace product restricted to the large eigenspace of $\Lambda$, without requiring the entry-size bounds on $\rcovstep$ and $\rcovsq$, as all terms in Lemma~\ref{lemma:largebucket} vanish except the diagonal ones in this case. We give details in Appendix~\ref{app:commute}.

\subsubsection{Gap bucket}

Next, we bound the contribution from the gap bucket.

\begin{lemma}[Gap bucket bound]
    \label{lemma:gapbucket}
For any $0 \leq t \leq 1$, 
\[
    \sum_{(i, j) \in \Gamma} (\rcovstep_{ij})^2 \nu(\lambda_i,\lambda_j) \leq \alpha \sum_{i = 1}^d \rcovstep_{ii} e^{-\lambda_i},
\]
and
\[
    \left| \sum_{(i, j) \in \Gamma} \rcovstep_{ij} \rcovsq_{ij} \nu(\lambda_i,\lambda_j) \right| \leq \alpha \sum_{i = 1}^d \rcovstep_{ii} e^{-\lambda_i}.
\]
\end{lemma}
\begin{proof}
First, observe $\lambda_i \leq 2K$ and $\lambda_j > 4K$ then $\nu(\lambda_i,\lambda_j) \leq \frac{e^{-\lambda_i}}{\lambda_j}$, as then $\lambda_i \leq \lambda_j / 2$. We start by proving the first claim. 
We have 
\[
    \sum_{(i, j) \in \Gamma} (\rcovstep_{ij})^2 \nu(\lambda_i,\lambda_j) = 2\sum_{\substack{(i, j): \lambda_i \leq 2K \\ \lambda_j > 4K}}(\rcovstep_{ij})^2 \nu(\lambda_i,\lambda_j) \leq \sum_{\substack{(i, j): \lambda_i \leq 2K \\ \lambda_j > 4K}} \frac{2e^{-\lambda_i}}{\lambda_j} (\rcovstep_{ij})^2 
\]
by the definition of $\Gamma$ and the above fact on $\nu(\lambda_i, \lambda_j)$, double counting pairs $(i, j)$. We will show
\[
    \sum_{j: \lambda_j > 4K} \frac{2e^{-\lambda_i}}{\lambda_j} (\rcovstep_{ij})^2 \leq \alpha \rcovstep_{ii} e^{-\lambda_i}
\]
for any $i$ such that $\lambda_i \leq 2K$: this implies the result after summing over all $i$ where $\lambda_i \leq 2K$, which is smaller than the right hand side of the desired bound. Let $\ell \defeq \{j: \lambda_j > 4K \}$. Consider the restriction of $\rcovstep$ onto $i \cup \ell$, which is clearly positive semidefinite:

\begin{equation*}
\begin{pmatrix}
\rcovstep_{ii} & \rcovstep_{i \ell} \\
\rcovstep_{\ell i} & \rcovstep_{\ell \ell}
\end{pmatrix}.
\end{equation*}

By taking the Schur complement of this matrix onto $\{ i \}$, we obtain $\rcovstep_{ii} - \rcovstep_{i \ell} \rcovstep_{\ell \ell}^{-1} \rcovstep_{i \ell}^\top \geq 0$. Now, note that by Lemma \ref{lemma:helper} part 2,
\[
    \rcovstep \preceq \frac{\alpha}{2} Q^\top \Psi Q  \preceq \frac{\alpha}{2} Q^\top(\Psi + t \rcovstep)Q = \frac{\alpha}{2} \Lambda.
\]
In addition, observe that $\rcovstep_{\ell \ell} \preceq \frac{\alpha}{2} \Lambda_{\ell \ell}$, as restrictions to blocks preserves the PSD ordering. Thus $\rcovstep_{\ell \ell}^{-1}\succeq \frac{2}{\alpha} \Lambda_{\ell \ell}^{-1}$, and so 
\[
\rcovstep_{ii} \geq \frac{2}{\alpha} \rcovstep_{i \ell} \Lambda_{\ell \ell}^{-1} \rcovstep_{i \ell}^\top = \frac{2}{\alpha} \sum_{j: \lambda_j > 4K} \frac{(\rcovstep_{ij})^2}{\lambda_j}.
\]
Rearranging this gives 
\[
\sum_{j: \lambda_j > 4K} \frac{2 e^{-\lambda_i}}{\lambda_j} (\rcovstep_{ij})^2 \leq \alpha \rcovstep_{ii} e^{-\lambda_i}
\]
as desired. As $0 \preceq \rcovsq \preceq \rcovstep$, the exact same chain of reasoning shows that  
\[
\sum_{j: \lambda_j > 4K} \frac{2 e^{-\lambda_i}}{\lambda_j} (\rcovsq_{ij})^2 \leq \alpha \rcovstep_{ii} e^{-\lambda_i}.
\]
By the AM-GM inequality, we then have
\[
\sum_{j: \lambda_j > 4K} \frac{2 e^{-\lambda_i}}{\lambda_j} |\rcovsq_{ij} \rcovstep_{ij}| \leq \sum_{j: \lambda_j > 4K} \frac{ e^{-\lambda_i}}{\lambda_j} ((\rcovstep_{ij})^2 + (\rcovsq_{ij})^2) \leq \alpha \rcovstep_{ii} e^{-\lambda_i},
\]
and summing over all $i$ where $\lambda_i \leq 2K$ and using triangle inequality implies the second claim.
\end{proof}

\subsubsection{Small bucket}
Finally, we turn our attention to bounding the sums restricted to the small bucket. 
\begin{lemma}[Small bucket bound]
    \label{lemma:smallbucket}
    For any $0 \leq t \leq 1$, 
    \[
        \sum_{(i, j) \in S} (\rcovstep_{ij})^2 \nu(\lambda_i,\lambda_j) \leq \frac{1}{2} \trprod{\rcovsq}{e^{-\Lambda}}
    \]
    and
    \[
        \left| \sum_{(i, j) \in S} \rcovstep_{ij} \rcovsq_{ij} \nu(\lambda_i,\lambda_j) \right| \leq \frac{1}{2} \trprod{\rcovsq}{e^{-\Lambda}}.
    \]
\end{lemma}
\begin{proof}
Let $\sml \defeq \{i: \lambda_i \leq 4K \}.$ By Lemma \ref{lemma:nubound},
\[
\sum_{(i, j) \in S} (\rcovstep_{ij})^2 \nu(\lambda_i,\lambda_j) \leq \sum_{(i, j) \in S} (\rcovstep_{ij})^2 \frac{e^{-\lambda_i} + e^{-\lambda_j}}{2}  = \Tr[(\rcovstep_{\sml \sml})^2 \exp(-\Lambda_{\sml \sml})].
\]
The equality follows after expanding the trace product as
\[
    \Tr[(\rcovstep_{\sml \sml})^2 \exp(-\Lambda_{\sml \sml})]= \sum_{i \in \sml} e^{-\lambda_i} ((\rcovstep_{\sml \sml})^2)_{ii} = \sum_{i \in \sml} \sum_{j \in \sml} (\rcovstep_{ij})^2 e^{-\lambda_i}
\]
and observing that the same sum with $e^{-\lambda_j}$ can be obtained by switching $i$ and $j$ in the sum. Similarly by applying the above reasoning to $\rcovsq$ and combining with AM-GM, 
\[
\sum_{i, j \in S} | \rcovstep_{ij} \rcovsq_{ij} | \nu(\lambda_i,\lambda_j)    \leq \frac{1}{2} \Tr[( (\rcovstep_{\sml \sml})^2 + (\rcovsq_{\sml \sml})^2) \exp(-\Lambda_{\sml \sml})].
\]
To conclude the proofs, we will show $(\rcovstep_{\sml \sml})^2 \preceq \frac{1}{2} \rcovsq_{\sml \sml}$ and $(\rcovsq_{\sml \sml})^2 \preceq \frac{1}{2} \rcovsq_{\sml \sml}$: the claims of the lemma follow by plugging these into the above bounds. For the first claim, we apply Lemma \ref{lem:matrixcs} with $A_i \defeq \alpha x_i (Q^\top C_i Q)_{\sml \sml}$, $c_i \defeq \delta_i$. Then,
\begin{equation*}
\begin{aligned}
\sum_{i=1}^d  A_i &= \alpha \sum_{i=1}^d x_i \left(Q^\top C_i Q\right)_{\sml \sml} \\
&\preceq \alpha \left(Q^\top (\Psi + tB) Q \right)_{\sml \sml} \\
&= \alpha \Lambda_{\sml \sml} \leq \frac{1}{2} I_{\sml \sml}
\end{aligned}
\end{equation*}
as $\alpha = \frac{1}{8K}$ and $\Lambda_{\sml \sml} \preceq 4K I_{\sml \sml}$ by definition. Thus Lemma \ref{lem:matrixcs} guarantees
\[
(\rcovstep_{\sml \sml})^2 = \left( \sum_{i=1}^d \delta_i A_i  \right)^2 \preceq \frac{1}{2} \sum_{i=1}^d \delta_i^2 A_i =  \frac{1}{2} \rcovsq.
\]
By instead choosing $c_i \defeq \delta_i^2$, the same argument gives us
\[
(\rcovsq_{\sml \sml})^2 \preceq \frac{1}{2} \sum_{i=1}^d \delta_i^4 A_i \preceq \frac{1}{2} \sum_{i=1}^d \delta_i^2 A_i \preceq \frac{1}{2} \rcovsq
\] 
where in the second inequality we used $0 \leq \delta_i \leq 1$. The claim follows.
\end{proof}

\subsubsection{Covering potential analysis}

With Lemmas \ref{lemma:largebucket}, \ref{lemma:gapbucket}, and \ref{lemma:smallbucket} in hand, we prove Lemma \ref{lemma:covering}, restated here for convenience.
\restateCov*

\begin{proof}
We begin by proving $-\iota'(t) \leq \epsilon \iota(t) + (nd\rho)^{-59}\te{-\Psi}.$ Note that by Lemma \ref{lemma:helper} $\covsq \preceq \frac{1}{2} \covstep$, and so 
\[
\iota(t) = \trprod{(\covstep - (1-t) \covsq)}{\exp(-\Psi - t \covstep)} \ge  \frac{1}{2} \trprod{\covstep}{\exp(-\Psi - t \covstep)}
\]
for $t \in [0,1]$. Now, recalling Lemma~\ref{lemma:iotaderiv}, 
\begin{align*}
    -\iota'(t) &=  \sum_{i, j \in [n_c]} (\rcovstep_{ij})^2 \nu(\lambda_i,\lambda_j) - (1 - t) \sum_{i, j \in [n_c]} \rcovstep_{ij} \rcovsq_{ij} \nu(\lambda_i,\lambda_j)  - \trprod{\rcovsq}{e^{-\Lambda}} \\
    &\leq  \sum_{i, j \in [n_c]} (\rcovstep_{ij})^2 \nu(\lambda_i,\lambda_j) + \left| \sum_{i, j \in [n_c]} \rcovstep_{ij} \rcovsq_{ij} \nu(\lambda_i,\lambda_j)\right|  - \trprod{\rcovsq}{e^{-\Lambda}}
\end{align*}
We now use Lemmas \ref{lemma:largebucket}, \ref{lemma:gapbucket}, and \ref{lemma:smallbucket} to bound the right hand side. The first sum becomes
\begin{align*}
    \sum_{i, j \in [n_c]} (\rcovstep_{ij})^2 \nu(\lambda_i,\lambda_j)  &=  \sum_{(i, j) \in S} (\rcovstep_{ij})^2 \nu(\lambda_i,\lambda_j) + \sum_{(i, j) \in \Gamma} (\rcovstep_{ij})^2 \nu(\lambda_i,\lambda_j) + \sum_{(i, j) \in L} (\rcovstep_{ij})^2 \nu(\lambda_i,\lambda_j) \\
    &\leq \frac{1}{2} \trprod{\rcovsq}{e^{-\Lambda}} + \alpha \trprod{\rcovstep}{e^{-\Lambda}} + (nd\rho)^{-60}\te{-\Psi}.
\end{align*} 
Similarly, 
\begin{align*}
\left| \sum_{i, j} \rcovstep_{ij} \rcovsq_{ij} \nu(\lambda_i,\lambda_j)\right| &\leq \left| \sum_{i, j \in S} \rcovstep_{ij} \rcovsq_{ij} \nu(\lambda_i,\lambda_j)\right| + \left| \sum_{i, j \in \Gamma} \rcovstep_{ij} \rcovsq_{ij} \nu(\lambda_i,\lambda_j)\right| + \left| \sum_{i, j \in L} \rcovstep_{ij} \rcovsq_{ij} \nu(\lambda_i,\lambda_j)\right| \\ 
&\leq \frac{1}{2} \trprod{\rcovsq}{e^{-\Lambda}} + \alpha \trprod{\rcovstep}{e^{-\Lambda}} + (nd\rho)^{-60}\te{-\Psi}.
\end{align*}
Plugging this into the bound on $-\iota'(t)$, we obtain (recalling $\alpha = 1/(8K) \le \eps/4$)
\begin{align*}
-\iota'(t) &\leq \trprod{\rcovsq}{e^{-\Lambda}} + 2 \alpha \trprod{\rcovstep}{e^{-\Lambda}} + 2 (nd\rho)^{-60}\te{-\Psi} - \trprod{\rcovsq}{e^{-\Lambda}} \\
&\leq  \frac{\epsilon}{2} \trprod{\covstep}{\exp(-\Psi- t\covstep)} + (nd\rho)^{-15}\te{-\Psi} \\
& \leq \epsilon \iota(t) + (nd\rho)^{-15}\te{-\Psi},
\end{align*}
as desired. Now, we apply Lemma \ref{lem:addgronwall}, the additive-error Gr\"onwall lemma, to $\iota$:
\[
\iota(1) \geq (1-\epsilon) \iota(0) - (nd\rho)^{-15}\te{-\Psi}.
\]
This implies
\begin{align*}
\te{-\Psi-\covstep} &\leq \te{-\Psi} - \iota(1) \\
&\leq \te{-\Psi} - (1-\epsilon) \iota(0) + (nd\rho)^{-15}\te{-\Psi} \\
&= \te{-\Psi} -(1-\epsilon) (\trprod{\covstep}{\exp(-\Psi)} - \trprod{\covsq}{\exp(-\Psi)}) + (nd\rho)^{-15}\te{-\Psi}.
\end{align*}
Finally, thus far we have shown
\[\Tr\exp(-\Psi - B) \le \te{-\Psi}\left(1 - (1 - \eps)\trprod{\covstep - \covsq}{\frac{\exp(-\Psi)}{\te{-\Psi}}} + (nd\rho)^{-15}\right).\]
Taking a logarithm of both sides, and using $\log(1 - x) < -x$ for $0 < x < 1$, concludes the proof.
\end{proof}

\subsection{Proof of Theorem~\ref{thm:mainclaimsdp}}

We recall that, in terms of the notation of this section, Theorem~\ref{thm:mainclaimsdp} asks us to show that
\[(1 - \eps)\log \Tr \exp(\Phi + A) + \log\Tr\exp(-\Psi - B) \le (1 - \eps)\log\Tr\exp(\Phi) + \log\Tr\exp(-\Psi) + (nd\rho)^{-15}.  \]
Combining a $1 - \eps$ multiple of Lemma~\ref{lemma:packing} and Lemma~\ref{lemma:covering} exactly shows this, up to proving that
\[\inprod{A+ G}{\frac{\exp(\Phi)}{\te{\Phi}}} - \inprod{B - H}{\frac{\exp(-\Psi)}{\te{-\Psi}}} \le 0.\]
The proof of this follows identically to the proof of Theorem~\ref{thm:mainclaimlp}, i.e. \eqref{eq:nonnegativelp}, \eqref{eq:nonnegativelp2} and onwards, where we recall the form of the gradient updates implies that this nonnegativity statement is equivalent to 
\[\alpha \sum_{i\in[d]} \delta_i x_i \left((1+ \delta_i)\pgradplaini - (1-\delta_i) \cgradplaini\right) \le 0,\]
and by Algorithm~\ref{alg:sdpalg}, for each $i$ either $\delta_i = 0$ or 
\begin{equation}\label{eq:keyclaim}(1+ \delta_i)\pgradplaini - (1-\delta_i) \cgradplaini = \frac{3 \pgradplaini \cgradplaini - (\pgradplaini)^2}{2\cgradplaini} - \frac{\pgradplaini + \cgradplaini}{2} = \frac{(\pgradplaini - \cgradplaini)^2}{2\cgradplaini} \leq 0.\end{equation}

\section{Convergence Analysis}
\label{sec:cleanup}

\subsection{Proofs of correctness} 
\label{ssec:correctness}
We closely follow the convergence analysis of \cite{MahoneyRWZ16}. In Lemma \ref{lem-sdp-infeas} we obtain a general certificate of infeasibility \eqref{eq:infeasible}, which we use to prove correctness of one termination condition in Lemma \ref{lem-coord-update}. Finally, using the approximate potential invariant due to Theorem~\ref{thm:mainclaimsdp}, Lemma~\ref{lem:terminationimpliesfeasibility} shows that the other termination condition certifies feasibility \eqref{eq:feasible}. Throughout this section, we define the potential function $f$ as in \eqref{eq:sdppotential}.

\begin{lemma}[Certificate of infeasibility]\label{lem-sdp-infeas} If there exist $Y \succeq 0$ and $Z \succeq 0$ such that \[(1 - \epsilon) \frac{\trprod{C_i}{Z}}{\Tr Z} < \frac{\trprod{P_i}{Y}}{\Tr Y}\; \forall i \in [d], \numberthis\label{ineq-sdp-infeas}\] then the SDP is infeasible.
\end{lemma}
\begin{proof} Suppose that \eqref{eq:infeasible} is feasible. Then there exists an $x\geq 0$ such that \[(1-\epsilon)  \frac{\trprod{\sum_{i\in[d]} x_i C_i}{Z}}{\Tr Z} \geq (1-\epsilon) \frac{\Tr Z}{\Tr Z} = 1-\epsilon, \] and \[ \frac{\trprod{\sum_{i\in[d]} x_i P_i}{Y}}{\Tr Y} \leq (1-\epsilon) \frac{\Tr Y}{\Tr Y} = 1 - \epsilon.\]  Combining with \eqref{ineq-sdp-infeas} obtains a contradiction. 
\end{proof}

We give a simple extension of Lemma~\ref{lem-sdp-infeas} that we will use in the later analysis.
\begin{lemma}\label{lem:averagepackinggrad}
	Suppose \eqref{eq:feasible} is feasible, and let $t_1, t_2, \ldots t_k$ be iterations in the algorithm. Then, there exists coordinate $i\in[d]$ such that \[\sum_{j \in [k]} \nabla_{P_i}^{(t_j)} \leq (1 - \epsilon) \sum_{j \in [k]} \nabla_{C_i}^{(t_j)}. \numberthis\label{existenceSmallPG}\] 
\end{lemma} 
\begin{proof}
	We proceed by contradiction; suppose \eqref{existenceSmallPG} fails for all $i \in [d]$. Recall that for any $t$,
	\[\pGradi = \frac{\trprod{P_i}{  Y^{(t)}}}{\Tr Y^{(t)}},\; \cGradi = \frac{\trprod{C_i}{Z^{(t)}}}{\Tr Z^{(t)}}.\]
	Now, define
	\[\bar{Y} \defeq \frac{1}{k}\sum_{j \in [k]} \frac{Y^{(t_j)}}{\Tr Y^{(t_j)}},\; \bar{Z} \defeq \frac{1}{k}\sum_{j \in [k]} \frac{Z^{(t_j)}}{\Tr Z^{(t_j)}}.\]
	and note $\Tr \bar{Y} = \Tr \bar{Z} = 1$. Then by this definition,
	\[\frac{1}{k} \sum_{j \in [k]} \nabla_{P_i}^{(t_j)} = \frac{\inprod{P_i}{\bar{Y}}}{\Tr \bar{Y}},\; \frac{1}{k} \sum_{j \in [k]} \nabla_{C_i}^{(t_j)} = \frac{\inprod{C_i}{\bar{Z}}}{\Tr \bar{Z}}.\]
	Therefore, applying the failure of \eqref{existenceSmallPG} for all $i \in [d]$ yields
	\[(1- \epsilon) \frac{\trprod{C_i}{\bar{Z}}}{\Tr \bar{Z}} < \frac{\trprod{P_i}{\bar{Y}}}{\Tr \bar{ Y}} \; \forall i \in [d],\] which by Lemma~\ref{lem-sdp-infeas} certifies infeasibility, a contradiction. 
\end{proof}

\begin{lemma}[Empty update implies infeasibility]\label{lem-coord-update} If \eqref{eq:feasible} is feasible, then for all $x\geq 0$, the set $W = \{i: \pGradplaini \leq \left(1-\frac{\eps}{2}\right) \cGradplaini\} \neq \emptyset$.
\end{lemma}
\begin{proof}Suppose there exists an $x\geq 0$ such that for all $i$, $\pGradplaini > \left(1-\frac{\eps}{2}\right) \cGradplaini$. By definition of $\pGradplain$ and $\cGradplain$, this means we have (where $Y$, $Z$ are defined as in Algorithm~\ref{alg:sdpalg} for this particular choice of $x$) \[(1- \epsilon ) \frac{\trprod{C_i}{Z}}{\Tr Z} \le \left(1 - \frac{\eps}{2}\right) \frac{\trprod{C_i}{Z}}{\Tr Z} < \frac{\trprod{P_i}{Y}}{\Tr Y}\; \forall i \in [d].\] By invoking Lemma \ref{lem-sdp-infeas}, we obtain a contradiction. 
\end{proof}
\begin{lemma}\label{lem:sdp-pfunc-noninc} 
During the execution of Algorithm~\ref{alg:sdpalg}, we always have $f(x^{(t)}) \leq 4\log n$. 
\end{lemma}
\begin{proof}
Initially we have $x_i^{(0)} = \frac{1}{d \lmax{P_i}}$, so $\Phi^{(0)} = \sum_{i \in [d]} x_i^{(0)} P_i \preceq I_{n_p}$; also, $\Psi^{(0)} = \sum_{i\in [d]} x_i^{(0)} C_i \succeq 0$. By the definition of the potential function \eqref{eq:sdppotential} and Fact~\ref{fact:softmaxmin}, we obtain 
\[f(x^{(0)}) = (1-\epsilon)  \smax(\Phi^{(0)}) - \smin(\Psi^{(0)}) \leq (1 - \epsilon)\left(\log n_p + \lmax{\Phi^{(0)}}\right) + \left(\log n_c - \lmin{\Psi^{(0)}}\right) \leq 3\log n.\] 
By repeatedly applying Theorem~\ref{thm:mainclaimsdp} for each iteration, and by the iteration bound in Section~\ref{ssec:iteration} and our assumption that $\epsilon \geq {(nd)}^{-3}$, over the course of the algorithm, the total accumulated error due to the ${(nd\rho)}^{-15} $ term is at most $\log n$. This proves the claim.
\end{proof}

\begin{lemma}\label{lem:terminationimpliesfeasibility}
Termination of Algorithm~\ref{alg:sdpalg} at line 16 certifies feasibility \eqref{eq:feasible}. Specifically, $x^{(t)} \ge 0$ satisfies \[ \lmax{\Phi^{(t)}} \leq (1 + 3\eps) \lmin{\Psi^{(t)}}.\]
\end{lemma}
\begin{proof} Nonnegativity of $x^{(t)}$ is due to $x^{(0)}\geq 0$ and monotonicity of $x^{(t)}$. Consider first termination due to $\lmax{\Phi^{(t)}} \ge K = 4\log(nd\rho)/\eps$. Since $x^{(t)} = x^{(t-1)}\circ(1 + \alpha \delta^{(t-1)})$, and $\alpha\delta^{(t - 1)} \le 1/16K$ entrywise, \[\Phi^{(t)} = \sum_{i\in [d]} x_i^{(t)} P_i \preceq \left( 1 + \frac{\epsilon}{64 \log (nd \rho)}\right) \sum_{i\in [d]} x^{(t-1)}_i P_i \preceq \left( 1 + \frac{\epsilon}{64}\right)\Phi^{(t - 1)}.\] 
The potential guarantee Lemma~\ref{lem:sdp-pfunc-noninc} and Fact~\ref{fact:softmaxmin} show
\[4\log n \ge f\left(x^{(t - 1)}\right) = (1 - \eps)\smax\left(\Phi^{(t - 1)}\right) - \smin\left(\Psi^{(t - 1)}\right) \ge (1 - \eps)\lmax{\Phi^{(t - 1)}} - \lmin{\Psi^{(t - 1)}}. \]
Finally, the conclusion of the lemma follows (recalling $\lmax{\Phi^{(t)}} \ge K$):
\begin{equation}
\label{eq:reuse}
\begin{aligned}
	\lmin{\Psi^{(t)}} \ge \lmin{\Psi^{(t - 1)}} &\geq (1- \epsilon)\lmax{\Phi^{(t-1)}}  - 4 \log n \\
	&\ge \frac{1-\epsilon}{1 + \frac{\eps}{64}}\lmax{\Phi^{(t)}} - \eps K \\
	&\ge (1 - 1.5\eps)\lmax{\Phi^{(t)}} - \eps\lmax{\Phi^{(t)}} \\
	&\ge \frac{1}{1 + 3\eps}\lmax{\Phi^{(t)}}.
\end{aligned} 
\end{equation}
Throughout we used $\eps \le \frac{1}{20}$. Next, consider termination due to $\lmin{\Psi^{(T)}} \geq K$. The chain of inequalities \eqref{eq:reuse}, stopping before the second-to-last inequality, concluded
\[\lmin{\Psi^{(t)}} \ge (1 - 1.5\eps)\lmax{\Phi^{(t)}} - \eps K.\] 
By the termination condition, this implies the desired
\[\lmax{\Phi^{(t)}} \le \frac{1+ \eps}{1 - 1.5\eps}\lmin{\Psi^{(t)}} \le (1 + 3\eps)\lmin{\Psi^{(t)}}. \]
\end{proof}
Finally, we remark that the guarantee of Lemma~\ref{lem:terminationimpliesfeasibility} is off by a factor of 3 from the desired guarantee. By setting $\eps \leftarrow \eps/3$ in the parameters, the feasibility check is correct, and whenever the algorithm returns ``infeasible'' for the smaller value of $\eps/3$, it is clear that \eqref{eq:infeasible} is also infeasible for $\eps$. This will affect the runtime only by constant factors.

\subsection{Iteration bound}
\label{ssec:iteration}
We now bound the iteration count of the algorithm. Our analysis follows the approach of \cite{Young01, MahoneyRWZ16} by defining a \emph{phase} as a set of iterations $t$ satisfying, for some integer $s$, \[ 2^s \leq \frac{\Tr Y^{(t)}}{\Tr Z^{(t)}}\leq 2^{s+1}.\]
By monotonicity of $x^{(t)}$ and as $\{P_i\}$ and $\{C_i\}$ are all positive semidefinite, monotonicity of $\Tr\exp$ implies $\frac{\Tr Y}{\Tr Z} = \frac{\Tr\exp(\Phi)}{\Tr\exp(-\Psi)}$ is non-decreasing, and therefore a phase is made up of consecutive iterations. 
\begin{lemma}\label{lem:numphases}
The total number of phases in Algorithm~\ref{alg:sdpalg} is $\Oh{\log (nd\rho)/\epsilon}$.
\end{lemma}
\begin{proof} If Algorithm~\ref{alg:sdpalg} has not terminated, then $\lmax{\Phi^{(t)}} \le K$ and $\lmin{\Psi^{(t)}} \le K$. Therefore, \[\frac{\Tr Y^{(t)}}{\Tr Z^{(t)}} = \frac{\Tr \exp(\Phi^{(t)})}{\Tr \exp(-\Psi^{(t)})} \leq \frac{n \exp(K)}{\exp(-K)} \leq n\exp(2K).\] Since each phase doubles the value of this ratio, and initially $Y^{(0)} \succeq 0$ and $Z^{(0)}\succeq 0$ we have \[\frac{\Tr Y^{(0)}}{\Tr Z^{(0)}} \geq 1,\] the above final bound implies that the number of phases is $\Oh{K}  = \Oh{\log (nd\rho)/\epsilon}$, as claimed. 
\end{proof}
To bound the number of iterations, we divide them into two types that we term ``fast'' and ``slow''\footnote{We use the terms ``fast/slow'' analogously to ``bad/good'' in \cite{MahoneyRWZ16}, as we believe it is less ambiguous.}. 
\begin{definition}
Iteration $t$ is ``slow'' if $\pGradi/\cGradi > \frac{1}{3}$ for all $i\in [d]$. Otherwise, $t$ is ``fast''. 
\end{definition}
This naming convention is because the update rule of Algorithm~\ref{alg:sdpalg} implies at least one coordinate increases by a constant factor on any fast iteration, allowing for a simple bound.
\begin{lemma}
\label{lem:numfastiters}
In a single phase, the number of fast iterations is at most $\Oh{\log^2(nd\rho)/\epsilon}$. 
\end{lemma}
\begin{proof}
Suppose iteration $t$ of the algorithm is fast, i.e. there is some $i\in [d]$ for which \[\frac{\pGradi}{\cGradi} = \frac{\trprod{P_i}{Y^{(t)}}}{\Tr Y^{(t)}} \cdot \frac{\Tr Z^{(t)}}{\trprod{C_i}{Z^{(t)}}} \leq 1/3.\]
Recalling that a phase is a sequence of consecutive iterations where $\frac{\Tr Y^{(t)}}{\Tr Z^{(t)}}$ changes by at most a factor of two, and the quantity $\frac{\trprod{P_i}{Y^{(t)}}}{\trprod{C_i}{Z^{(t)}}}$ is monotonic, we conclude that the iterations $t^\prime < t$ in the same phase have $\mathbf{\nabla}_{P_i}^{(t^\prime)}/\mathbf{\nabla}_{C_i}^{(t^\prime)} \leq 2/3$. This implies, from the definition of $\delta$ in the algorithm, that \[\delta_i^{(t^\prime)} = \frac{1}{2}\left( 1 - \frac{\mathbf{\nabla}_{P_i}^{(t^\prime)}}{\mathbf{\nabla}_{C_i}^{(t^\prime)}}\right) \geq 1/6\]
for all these iterations. Thus, the algorithm updates $x_i$ multiplicatively by a factor of at least 
\[1  + \frac{\alpha}{6} \ge \exp\left(\frac{\eps}{200\log(nd\rho)}\right)\] 
for all these fast iterations of the phase. Note that whenever $x_i^{(t)} \ge \frac{K}{\lambda_{\max}(P_i)}$,
\[\lmax{\sum_{i \in [d]} x_i^{(t)} P_i} \ge \lmax{x_i^{(t)} P_i} \ge K,\]
i.e. the algorithm terminates. At initialization, $x_i^{(0)} = \frac{1}{d \lmax{P_i}}$. Thus, the number of fast iterations in this phase is bounded by the desired
\[\frac{200\log(nd\rho)}{\eps}\log (Kd) = O\left(\frac{\log^2(nd\rho)}{\eps}\right).\]
\end{proof}
The above lemma shows that after the first $\Oh{\log^2(nd\rho)/\epsilon}$ iterations of a phase, all subsequent iterations in the phase must be slow, the number of which we now bound. 
\begin{lemma}\label{lem:numSlowIters}
In a single phase, the number of slow iterations is at most $\Oh{\log^2 (nd\rho) /\epsilon^2}$. 
\end{lemma}
\begin{proof}
Label the consecutive slow iterations in the phase $1, 2, \ldots T$ for simplicity, and let $i\in [d]$ satisfy \eqref{existenceSmallPG} for all but the last slow iteration. Also let $\ell = \pGradiniti$ and $u = \cGradiniti$. By monotonicity of $\trprod{P_i}{Y}$ (resp. $\trprod{C_i}{Z}$) and the fact that $\Tr Y$ (resp. $\Tr Z$) does not change by more than a factor of two in a phase, we have \[\pGradi =  \frac{\trprod{P_i}{Y^{(t)}}}{\Tr Y^{(t)}} \geq \ell/2, \, \forall t\in [T], \] and \[\cGradi =  \frac{\trprod{C_i}{Z^{(t)}}}{\Tr Z^{(t)}} \leq 2 u, \, \forall t\in [T].\] By definition of ``slow iteration'', we have $\ell \geq u/3$, and also (accounting for when $\delta_i^{(t)} = 0$)
\begin{equation}\label{eq:deltabound}\delta_i^{(t)} \ge \frac{\left(1 - \frac{\eps}{2}\right)\cGradi - \pGradi}{2\cGradi}. \end{equation}
Therefore, 
\begin{align*}
x_i^{(T)} &\geq x_i^{(1)} \exp\left( \frac{\alpha \sum_{t\in [T - 1]} \delta_i^{(t)}}{2} \right) \\
			&\ge x_i^{(1)} \exp\left( \frac{\alpha}{4} \sum_{t\in[T - 1]} \frac{\left(1 - \frac{\eps}{2}\right)\cGradi - \pGradi}{\cGradi} \right)\\
			&\geq x_i^{(1)} \exp \left( \frac{\alpha}{4} \sum_{t\in [T - 1]} \left(\frac{\left(1 - \frac{\eps}{2}\right)\cGradi -  \pGradi}{2u} \right)   \right).
\end{align*} Because $i$ satisfies \eqref{existenceSmallPG}, we have 
\begin{align*}
\sum_{t \in [T - 1]} \left(1 - \frac{\eps}{2}\right) \cGradi - \pGradi &\geq \frac{\eps}{2} \sum_{t\in [T - 1]} \cGradi \\
									&\geq \frac{\epsilon}{2(1-\epsilon)} \sum_{t\in [T - 1]} \pGradi\\
									&\geq \frac{\eps (T - 1)u}{12}
\end{align*}
Therefore, $x_i^{(T)} \ge x_i^{(1)}\exp(\eps \alpha (T - 1)/96)$. In the proof of Lemma~\ref{lem:numfastiters}, we showed that whenever 
\[\frac{\eps\alpha (T - 1)}{96} \ge \log (Kd),\]
the algorithm terminates. As $\alpha = \frac{\eps}{32\log(nd\rho)}$, this gives the desired iteration bound.
\end{proof}
\subsection{Implementation details}
\label{ssec:implementation}
Every iteration of our algorithm may be implemented in time $O(\nnz(\allP) + \nnz(\allC) + n^\omega)$, i.e. linear time with an additional matrix multiplication time overhead. The goal of this section is to categorize a wide family of instances where the multiplication time overhead may be removed, and the iteration complexity may be brought to nearly-linear. We use the following assumption.

\begin{assumption}
	\label{assume:matexp}
	For a positive SDP with packing matrices $\{P_i\}_{i \in [d]}$, covering matrices $\{C_i\}_{i \in [d]}$, and variables $x_i \geq 0$ arising from an execution of Algorithm \ref{alg:sdpalg}, assume matrix-vector multiplication access to symmetric matrices $A$ and $B$ satisfying
	\begin{equation}
	\begin{aligned}
	\exp\left(\half\sum_{i \in [d]} x_i P_i\right) &\preceq A \preceq \left(1+\frac{\epsilon}{70}\right) \exp\left(\half\sum_{i \in [d]} x_i P_i\right) \\
	\exp\left(-\half\sum_{i \in [d]} x_i C_i\right) &\preceq B \preceq \left(1+\frac{\epsilon}{70}\right) \exp\left(-\half\sum_{i \in [d]} x_i C_i\right) + e^{-20K} I_{n_c}
	\end{aligned}
	\end{equation}
	in time $\tpexp$ and $\tcexp$ respectively, such that $A$ and $\sum_{i \in [d]} x_i P_i$ commute, and $B$ and $\sum_{i \in [d]} x_i C_i$ commute.
\end{assumption}

Note that Assumption~\ref{assume:matexp} may be satisfied with $\tpexp = O(\nnz(\allP) + n_p^\omega)$ and $\tcexp = O(\nnz(\allC) + n_c^\omega)$ and logarithmic parallel depth due to e.g. \cite{PanC99}. As argued in \ref{ssec:hardness}, it is unlikely that this is improvable for general positive SDP instances, barring a breakthrough in our understanding of the hardness of many linear-algebraic primitives. In Section~\ref{ssec:poly} we discuss efficient implementation of Assumption~\ref{assume:matexp} under additional structural assumptions. We then discuss the implications of Assumption~\ref{assume:matexp} for our algorithm in Sections~\ref{ssec:termination} and~\ref{ssec:JL}. Typically, the alternative implementations lose small constant approximation factors in their guarantees. For simplicity, we defer discussion of the validity of our analysis under these relaxed guarantees in Appendix~\ref{app:approximate} and discuss only implementation details here. 

\subsubsection{Assumption~\ref{assume:matexp} for structured instances}
\label{ssec:poly}

For all implementations in this section, commutativity follows immediately from polynomials, scalar multiplication, inversions, and addition with the identity preserving commuting.

\paragraph{$\tpexp$ in nearly-linear time.}
We first show that the oracle for the packing matrices in Assumption~\ref{assume:matexp} can be unconditionally implemented efficiently, making use of the following fact from the literature on polynomial approximation \cite{OrecchiaSV12, SachdevaV14}; we remark that because our error tolerance is multiplicative, the square root savings of Chebyshev polynomials do not apply.

\begin{lemma}[Polynomial approximation for exponential]
	\label{lemma:exp_poly}
	Let $M \in \Sym^n_{\geq 0}$ satisfy $M \preceq b I$ for some $b > 0$. Let $\delta$ be a parameter such that $\delta^{-1} < b$. There exists a polynomial $p$ of degree $O(b)$ such that $\exp(M) \preceq p(M) \preceq (1 + \delta)\exp(M)$. 
\end{lemma}

\begin{lemma}[Packing oracle implementation]
	$\tpexp = O\left(\log(n d \rho) \epsilon^{-1} \nnz(\allP)\right)$.
\end{lemma}
\begin{proof}
	Observe that by the termination condition (tolerant to the approximation guarantee to be developed in Section~\ref{ssec:termination}), every $x$ encountered during the algorithm is such that $0 \preceq \sum_{i \in [d]} x_i P_i \preceq 2K I_{n_p}$, where $K = O(\log(nd\rho)/\eps)$. By applying Lemma \ref{lemma:exp_poly} with $b = K$ and $\delta = \Omega((nd)^{-3})$, we obtain a polynomial $p$ of degree $O(\log(n d \rho) \epsilon^{-1} )$ such that $p\left(\half\sum_{i \in [d]} x_i P_i\right)$ meets the guarantees required by Assumption~\ref{assume:matexp}. Matrix-vector multiplication can be implemented in time proportional to the degree of the polynomial approximation. 
\end{proof}

Unfortunately, as alluded to earlier, we cannot guarantee unconditionally fast implementations for $\tcexp$. We now characterize several structured instances admitting improved implementations.

\paragraph{$\tcexp$ under multiplicative bounds.} For instances where we can guarantee that throughout the algorithm, $0 \preceq \sum_{i \in [d]} x_i C_i \preceq \beta KI_{n_c}$ for some mild parameter $\beta$, we can use the developments of the prior section to straightforwardly conclude the following:
\begin{lemma}[Bounded covering oracle implementation]
	$\tcexp = O\left(\beta\log(n d \rho) \epsilon^{-1} \nnz(\allC)\right)$.
\end{lemma}
A typical application admitting this type of structure is when the covering matrices are a multiple of the packing matrices, i.e. $C_i = \beta P_i$ for all $i \in [d]$. This is typical of positive SDP instances arising from sparsification \cite{LeeS17, JambulapatiSS18} and robust statistics \cite{ChengG18}.

\paragraph{$\tcexp$ with linear system solves.} Assuming access to a linear system solver for matrices of the form $\alpha I_{n_c} + \sum_{i \in [d]} x_i C_i$ for nonnegative $\alpha$ and $x$, we are able to efficiently implement the oracle required by Assumption~\ref{assume:matexp}. We will use the following fact \cite{OrecchiaSV12, SachdevaV14}.

\begin{lemma}
	\label{lemma:exp_poly2}
	Let $M \in \Sym^n_{\geq 0}$. There is a polynomial $p$ of degree $d = O(\log(n \delta^{-1}))$ such that 
	\[
	\exp(-M) \preceq p \left(\left(I_n + \frac{M}{d}\right)^{-1} \right) \preceq \exp(-M) + \delta I_n.
	\]
\end{lemma}

\begin{lemma}[Covering oracle implementation with solver]
	Assume we can solve linear systems in $\alpha I_{n_c} + \sum_{i \in [d]} x_i C_i$ for any nonnegative $\alpha$, $x$ in time $\mathcal{S}_{\textup{cov}}$. Then
	$\tcexp = O(\nnz(\allC) + K\mathcal{S}_{\textup{cov}})$.
\end{lemma}
\begin{proof}
	By applying Lemma \ref{lemma:exp_poly2} with $\delta = e^{-10K}$, the result follows immediately. 
\end{proof}

In particular, note that under Assumption~\ref{assume:inverse}, this implies that we can guarantee $\tcexp$ in nearly-linear time. Families of covering matrices for which Assumption~\ref{assume:inverse} holds include the diagonal covering matrices originally considered in \cite{JainY12}, symmetric diagonally dominant matrices (e.g. graph Laplacians) \cite{SpielmanT14}, and symmetric M-matrices \cite{AhmadinejadJSS19}.

\subsubsection{Certifying termination}
\label{ssec:termination}
In this section we demonstrate how to efficiently check termination in line 4 of Algorithm~\ref{alg:sdpalg}. First, recall that the algorithm terminates at iterate $x^{(t)}$ if
\[
\lmax{ \sum_{i \in [d]} x_i^{(t)} P_i } \geq K.
\] 
Thus to check termination it suffices to compute the top eigenvalue of $ \sum_{i \in [d]} x_i^{(t)} P_i$. The Lanczos method \cite{KuczynskiW92} allows us to do so to sufficient accuracy in time nearly-linear in $\nnz(\allP)$. We remark that the classical power iteration method also suffices for our purposes, with a slightly worse complexity.
\begin{fact}[Theorem 3.2 of \cite{KuczynskiW92}]
	Let $A \in \PSDSet^n$. For any $\eta > 0$, there is an algorithm that uses $O \left(\log(nd)\eta^{-1/2} \right)$ matrix-vector products with the matrix $A$, which outputs a unit vector $u \in \R^n$ so that $u^\top A u \geq (1 - \eta) \lmax{A}$, with probability at least $1 - 1/(nd)^{15}$.
\end{fact}
As an immediate corollary of this, by setting $\eta = 1/2$, we may check the termination condition up to a multiplicative factor of $2$ in time $O(\nnz(\allP)\log (nd))$ with high probability. Therefore, with high probability, we can terminate the algorithm when $\lmax{ \sum_{i \in [d]} x_i^{(t)} P_i } \in [K, 2K]$, by terminating whenever the Lanczos method says the approximate largest eigenvalue is at least $K$. Next, the algorithm also terminates on iterate $x^{(t)}$ if
\[
\lmin{ \sum_{i \in[d]} x_i^{(t)} C_i } \geq K.
\] 
The main idea is to use the Lanczos method to approximate the top eigenvalue of $\exp(-\sum_{i \in [d]} x_i^{(t)} C_i)$, which gives a tight approximation to the smallest eigenvalue of $\sum_{i \in[d]} x_i^{(t)} C_i$. More formally, recall that in Assumption~\ref{assume:matexp}, we have approximate implementations of matrix-vector products against $\exp(-\half\sum_{i \in [d]} x_i^{(t)} C_i)$, for a deterministic matrix $B$, with the guarantee
\[\exp\left(-\half\sum_{i \in [d]} x_i C_i\right) \preceq B \preceq \left(1+\frac{\epsilon}{70}\right) \exp\left(-\half\sum_{i \in [d]} x_i C_i\right) + e^{-20K} I_{n_c}.\]
Thus, using the Lanczos method with $B$ allows us to compute a $\eta = 1/2$-approximate top eigenvalue of $B$ in $O(\log(nd))$ matrix-vector products, which, while $\sum_{i \in[d]} x_i^{(t)} C_i \preceq 2KI_{n_c}$ holds, also is a $1/e$-approximate top eigenvalue for $\exp(-\half \sum_{i \in [d]} x_i^{(t)} C_i)$, as the error induced by the multiplicative and additive error guarantees are negligible. Taking a logarithm and scaling allows us to compute an eigenvalue that is within an additive $2$ to the smallest eigenvalue of $\sum_{i \in [d]} x_i^{(t)} C_i$, so we may again terminate when the smallest eigenvalue is in the range $[K, 2K]$ with high probability.

\subsubsection{Efficient trace products}
\label{ssec:JL}
In this section, we discuss the implementation of the trace products of $\{P_i\}_{i \in [d]}$ and $\{C_i\}_{i \in [d]}$ against $Y^{(t)}$ and $Z^{(t)}$ required by lines 7 and 8 of Algorithm~\ref{alg:sdpalg}. We show that this approximate implementation does not affect the runtime by more than constants in Appendix~\ref{app:approximate}. We require the following result, which follows from the classical Johnson-Lindenstrauss lemma (see e.g. \cite{DasguptaG03}).

\begin{lemma}[Johnson-Lindenstrauss]
	\label{lem:JL}
Let $A,B \in \Sym^n_{\ge 0}$. Let $k = O(\log \delta^{-1}/\epsilon^2)$ for some $\epsilon \in (0, 1/3)$, and let $Q \in \R^{n \times k}$ be a matrix where every entry is sampled from $\mathcal{N}(0,1/k)$. Then with probability at least $1 - \delta$,
\[\Tr[A^2 B] \le \Tr\left[Q^\top A B A Q\right]\le \left(1 + \frac{\eps}{70}\right)\Tr[A^2 B] .
\]
\end{lemma}

\begin{lemma}[Efficient trace products]
\label{lemma:stepworacle}
For $\pGradi$ and $\cGradi$ defined in lines 7-8 of Algorithm~\ref{alg:sdpalg}, given Assumption~\ref{assume:matexp}, each iteration, we can output values
\[\widetilde{\pGradi} \in \left[ \left(1 - \frac{\eps}{20}\right) \pGradi, \left(1 + \frac{\eps}{20}\right) \pGradi \right],\; \widetilde{\cGradi} \in \left[\left(1 - \frac{\eps}{20}\right) \cGradi, \left(1 + \frac{\eps}{20}\right)\cGradi + e^{-10K} \Tr[C_i] \right]\] 
in time $O\left(\left(\nnz(\allP) + \nnz(\allC) + \tpexp + \tcexp\right) \log (nd) \epsilon^{-2}\right)$, with probability of failure at most $(nd)^{-15}$.
\end{lemma}
\begin{proof}
Fix some $i \in [d]$ for this discussion, as they all follow similarly. The algorithm requires
\begin{equation}\label{eq:required}\pGradi = \frac{\trprod{P_i}{\exp\left(\Pxcurr\right)}}{\Tr\left[\exp\left(\Pxcurr\right)\right]},\;\cGradi = \frac{\trprod{C_i}{\exp\left(-\Cxcurr\right)}}{\Tr\left[\exp\left(-\Cxcurr\right)\right]},\end{equation}
in these lines. We claim that computing the numerators and denominators in \eqref{eq:required} via the matrices $A$ and $B$ given by Assumption \ref{assume:matexp} suffice. In particular, consider estimating
\[\widetilde{\pgradi} = \frac{\Tr[Q^\top A P_i A Q]}{\Tr[Q^\top A^2 Q]},\; \widetilde{\cgradi} = \frac{\Tr[Q^\top B C_i B Q]}{\Tr[Q^\top B^2 Q]},\]
where $Q$ is the matrix given by Lemma~\ref{lem:JL}. We can implement these trace products by computing the values $Aq$, $Bq$ for each column $q$ of the matrix $Q$ in time $O(\tpexp + \tcexp)$, and then computing the respective quadratic forms in time $O(\nnz(\allP) + \nnz(\allC))$. Summing this runtime over all columns of $Q$ yields the runtime claim. We now prove correctness. 

First, for correctness of $\widetilde{\pGradi}$, note that for commuting matrices $M_1 \preceq M_2$, we also have $M_1^2 \preceq M_2^2$. Therefore, $A^2$ approximates $\exp(\Pxcurr)$ to a $(1 + \eps/70)^2$ factor; then, combining the errors incurred by the trace products through Lemma~\ref{lem:JL}, for both the numerator and the denominator, the overall multiplicative error is at most $(1 + \eps/70)^6 \le 1 + \eps/10$. Turning this one-sided error guarantee into a two-sided guarantee yields the final claim.
 
Next, note that because $\exp(-\half\Cxcurr) \preceq I_{n_c}$, we have by commutativity that 
\[\exp\left(-\Cxcurr\right) \preceq B^2 \preceq \left(1 + \frac{\eps}{70}\right)^2\exp\left(-\Cxcurr\right) + 3e^{-20K} I_{n_c}.\]
Thus, the numerator of $\cGradi$ can be estimated to the desired multiplicative accuracy, up to the additive factor $3e^{-20K}\Tr[C_i]$. Finally, we are able to obtain a multiplicative guarantee on the denominator, because as the algorithm has not terminated, $\Tr\left[\exp\left(-\Cxcurr\right)\right] \ge \exp(-2K)$. Therefore, the additive error incurred induces only a negligible multiplicative error in the denominator estimation. Combining with the multiplicative error due to Lemma~\ref{lem:JL} yields the final guarantee.
\end{proof}

\subsection{Runtime}
We comment on the runtime claims made in Section~\ref{ssec:results}, and how they follow from Sections~\ref{ssec:correctness} and \ref{ssec:iteration}. The complexities per iteration claimed follow generically from the discussion of Section~\ref{ssec:implementation}, where we show tolerance of the algorithm to approximate guarantees in Appendix~\ref{app:approximate}. Therefore, combining Lemmas~\ref{lem:numphases},~\ref{lem:numfastiters}, and~\ref{lem:numSlowIters} we obtain the general bound of $O(\log^3(nd\rho)/\eps^3)$ on the number of iterations. We now briefly discuss the runtime improvements claimed by specialization to more structured instances:
\begin{itemize}
	\item For pure packing and covering, the improved runtime results from an analogous discussion in \cite{MahoneyRWZ16}, where they note that because either the covering or packing gradient is constant, so we are able to aggregate all slow iterations through phases and prove that they are bounded in number by $O(\log^2(nd\rho)/\eps^2)$. In the case of pure covering, the dominant term in the runtime thus comes from combining Lemmas~\ref{lem:numphases} and~\ref{lem:numfastiters}. For the case of pure packing, we can improve the iteration count by another logarithmic factor, and remove the width dependence, as described in the following item.
	\item For commutative covering constraints, our only dependence on $\rho$ through the entire analysis was in bounding the size of the contribution of the large bucket, in Section~\ref{ssec:largegap}, an assumption that may be removed when the constraints commute. In this case, it is clear that by appropriately modifying the definition of the threshold $K$ in the algorithm we may generically remove the dependence on $\rho$ for all parts of the complexity of the algorithm.
\end{itemize}

We defer a more formal discussion of these runtime improvements to Appendix~\ref{app:pure}.
\section{Future work}
\label{sec:futurework}

\paragraph{Necessity of width.}
We consider it interesting to formally explore if our polylogarithmic dependence on the quantity $\rho$, a natural notion of width for a mixed positive SDP instance, is necessary. In particular, it appears to us that without an algorithm which is able to efficiently compute projections onto a subset of eigenvalues (which in general requires eigendecomposition), and an analysis amenable to such a procedure, such a dependence is in general necessary for analyses akin to the one explored in this paper in, e.g. characterizing the error incurred by using an implicit truncation to perform a Taylor expansion of a matrix exponential.

\paragraph{Removing the need for phases.}
It was mentioned as an open problem in \cite{MahoneyRWZ16} to analyze their algorithm, or its variants, in a framework which does not involve explicit phases (cf. Section~\ref{sec:cleanup}). Because the best parallel solvers for pure packing or pure covering instances have an iteration complexity of roughly $\tOh{\eps^{-2}}$, it is natural to conjecture that the parallel complexity of mixed solvers may also be improved.

\paragraph{Sketching machinery.}
Recently, a variety of works have proposed improving the dependence on the rank of the low-rank approximation required to efficiently approximate matrix exponential-vector products. In particular, a line including \cite{KalaiV05, BaesBN13, GarberHM15, Allen-ZhuL17, CarmonDST19} studied the behavior of approximate (width-dependent) semidefinite program solvers whose iterations were implemented with sketches of smaller rank than the $\tOh{\eps^{-2}}$ typical to Johnson-Lindenstrauss projections, culminating in the rank-1 sketch of \cite{CarmonDST19}. 

Nonetheless, there remains difficulty in using rank-deficient sketches to implement steps in a regret-based framework for the positive SDP setting, as discussed in \cite{CarmonDST19}; the main obstacle is due to the fact that the gain matrices used to inform the update in an iteration is \emph{not} independent of the randomness used to determine the iterate, due to the presence of ``update buckets'' such as our $W^{(t)}$ (cf. Algorithm~\ref{alg:sdpalg}). We leave the possibility of relaxed independence assumptions, analyzing the correlation between the randomness used and properties of the iterate, or an alternative framework as interesting future work. We hope that the additional matrix analysis tools provided in this work may prove useful in analyzing e.g. the stability of iterations, towards these efforts.

\paragraph{Parallel acceleration.}
An important unresolved problem in the positive linear programming literature is the possibility of an \emph{accelerated} algorithm for even pure packing or covering LP instances, akin to the method developed in \cite{Nesterov83} for general convex programming, and in \cite{Nemirovski04} for e.g. its minimax generalization. This problem was partially resolved by the work of \cite{ZhuO15B}, which obtained an algorithm with total work proportional to $\eps^{-1}$, but dimension-dependent parallel depth due to the use of a coordinate scheme. 

Towards the goal of accelerated first-order methods for structured linear programming instances, the important breakthrough work of \cite{Sherman17} gave an accelerated method with convergence rate $\eps^{-1}$ for bilinear minimax problems in $\ell_\infty$-operator norm bounded matrices. Crucially, this work designed a method which utilized stronger interplay between primal and dual iterates than is possible in a primal-only method, under certain algorithmic restrictions (e.g. strong convexity). Because the fastest positive linear programming algorithms \cite{ZhuO15A, WangMMR15, MahoneyRWZ16} are primal-only, and the analyses asymmetric, we hope that the tools developed in this paper for viewing the covering side more symmetrically, e.g. without explicit pruning, may help lead the way to a new perspective towards primal-dual algorithms for positive LPs and SDPs.

\paragraph{Dichotomy characterization for efficient iterations.}
We find it interesting to characterize precisely the families of mixed positive SDP problems for which it is possible to implement efficient iterations, following the discussion in Section~\ref{ssec:hardness}. We note that in recent works a similar type of dichotomy characterization was achieved \cite{KyngZ17, KyngWZ20}; we believe mixed positive SDP solvers are a strong algorithmic primitive, and believe such a characterization to be fundamental to the development of this area. 	
	\subsection*{Acknowledgments}
	JL would like to thank Cameron Musco, Christopher Musco, and Zeyuan Allen-Zhu for helpful conversations. SP would like to thank Kuikui Liu for helpful conversations. AJ and KT would like to thank Kiran Shiragur for helpful conversations, and Aaron Sidford for bringing this problem to our attention, as well as his continued encouragement through this project.
	\newpage
	\bibliographystyle{alpha}	
	\bibliography{mixed-positive}
	\newpage
	\begin{appendix}	

\section{Reduction from optimization to decision}
\label{app:reduction}

In this section, we discuss a reduction from the optimization problem \eqref{eq:possdpdef} to the decision version of the problem \eqref{eq:feasible}, \eqref{eq:infeasible} considered in the main body of this paper. We restate the optimization problem here for convenience:
\begin{equation*}\begin{aligned}\min \mu \text{ subject to } \sum_i P_i x_i \preceq \mu I_{n_p}, \sum_i C_i x_i \succeq I_{n_c}, x \geq 0 , \\
\text{for sets of matrices } \left\{P_i \in \PSDSet^{n_p}\right\}_{i \in [d]}, \left\{C_i \in \PSDSet^{n_c}\right\}_{i \in [d]}.\end{aligned}\end{equation*}
 We then discuss the implications of this reduction for the runtime of the full algorithm for solving \eqref{eq:possdpdef}. As discussed in Section~\ref{ssec:width}, we first assume that we have a range $[\mu_{\text{lower}}, \mu_{\text{upper}}]$ containing $\mu_{\mathrm{OPT}}$, the value of~\eqref{eq:possdpdef}; we then demonstrate how to derive simple bounds $\mu_{\text{lower}}, \mu_{\text{upper}}$ from an optimization instance \eqref{eq:possdpdef}.

The reduction from the optimization problem to the decision problem is straightforward: we iteratively binary search over the range of $\mu$, and we check, by appropriately rescaling $\{C_i\}_{i\in[d]}$ and $\{P_i\}_{i\in[d]}$ whether or not our guess of $\mu$ is a valid solution, using our decision oracle. Formally:
\begin{lemma}
\label{lem:reduction}
Assume an oracle $\mathcal{O}$, which given $\{\tilde{C}_i\}_{i = 1}^d, \{\tilde{P}_i\}_{i = 1}^d$, runs in time $\mathsf{T}_{\text{decision}}$, and returns a solution to the corresponding decision problem \eqref{eq:feasible}, \eqref{eq:infeasible} with accuracy $\eps/3$. Then, there is an algorithm $\mathsf{MixedPositiveSolver}$ which takes oracle $\mathcal{O}$ and an instance \eqref{eq:possdpdef} given by $\{C_i\}_{i = 1}^d, \{P_i\}_{i = 1}^d$ as input, as well as a range $[\mu_{\text{lower}}, \mu_{\text{upper}}]$ containing $\mu_{\mathrm{OPT}}$ the value of \eqref{eq:possdpdef}, and outputs $x \in \R^d$ and $\mu > 0$ so that 
\[
\sum_i P_i x_i \preceq \mu I_{n}, \sum_i C_i x_i \succeq I_{n}, x \geq 0 \; ,
\]
with $\mu \le (1 + \eps) \mu_{\mathrm{OPT}}$.
Moreover, the algorithm runs in time $O (\mathsf{T}_{\text{decision}} \cdot  (\log\log(\mu_{\text{upper}}/\mu_{\text{lower}}) + \log \eps^{-1}))$.
\end{lemma}
\begin{proof}
First, we divide the range $[\mu_{\text{lower}}, \mu_{\text{upper}}]$ into multiples of $1 + \eps$ starting from $\mu_{\text{lower}}$, i.e. the set of values $\mu_{\text{lower}} \cdot (1 + \eps)^k$ for $k \in [K]$, and  
\[K = O\left(\frac{1}{\eps} \cdot \log\left(\frac{\mu_{\text{upper}}}{\mu_{\text{lower}}}\right)\right).\]
The reduction is to iteratively guess values $\mu = \mu_{\text{lower}} \cdot (1 + \eps)^k$ for some value $k \in [K]$, and perform a binary search over $[K]$, giving the desired runtime. For each guess $k$, we can certify whether $\mu_{\mathrm{OPT}} \ge \mu$ by testing feasibility of the set $\{C_i\}_{i = 1}^d, \{(1 - \eps/3)\mu^{-1} P_i\}_{i = 1}^d$ by using the oracle $\mathcal{O}$ in time $\mathsf{T}_{\text{decision}}$. By the fact that $(1 + \eps/3) / (1 - \eps/3) \le 1 + \eps$, there is a unique smallest value $k$ for which the binary search will return feasible, and correctness of the binary search follows from monotonicity of feasibility. Moreover, by the guarantee of \eqref{eq:infeasible}, the oracle will certify the infeasibility of $\mu$ whenever it is not a valid solution to \eqref{eq:possdpdef}, which shows that the value returned will always be correct.
\end{proof}

Next, given an instance \eqref{eq:possdpdef}, we show how to derive simple upper and lower bounds on the optimal value $\mu_{\mathrm{OPT}}$.

\begin{lemma}
Given an instance \eqref{eq:possdpdef}, we have $\mu_{\mathrm{OPT}} \in [\mu_{\text{lower}}, \mu_{\text{upper}}]$, where
\[\mu_{\text{upper}} = \min_{i \in [d]} \frac{\lambda_{\max}(P_i)}{\lambda_{\min}(C_i)},\; \mu_{\text{lower}} = \frac{\min_{i \in [d]} \Tr(P_i)}{\max_{i \in [d]} \Tr(C_i)} \cdot \frac{n_c}{n_p} \ge \frac{\min_{i \in [d]} \lambda_{\max}(P_i)}{n_p \max_{i \in [d]} \lambda_{\max}(C_i)}.\]
\end{lemma}
\begin{proof}
Correctness of the upper bound follows by its feasibility against choosing any 1-sparse $x$ with $x_i = \frac{1}{\lambda_{\min}(C_i)}$, yielding $\mu_{\text{OPT}} \le \frac{\lambda_{\max}(P_i)}{\lambda_{\min}(C_i)}$; taking the smallest value over $i \in [d]$ gives $\mu_{\text{upper}}$.

We now show correctness of the lower bound. First, we lower bound $\norm{x}_1$ for any $x \ge 0$ satisfying
\[\sum_{i \in [d]} C_i x_i \succeq I_{n_c}.\]
Note that this implies $\Tr\left(\sum_{i \in [d]} C_i x_i\right) \ge n_c$. Thus, we conclude via the following simple inequality
\[\Tr\left(\sum_{i \in [d]} C_i x_i\right) \le \norm{x}_1 \max_{i \in [d]} \Tr(C_i) \Rightarrow \norm{x}_1 \ge \frac{n_c}{\max_{i \in [d]} \Tr(C_i)}. \]
Next, using a similar strategy we see that for any $\lambda \ge 0$ satisfying
\[\sum_{i \in [d]} P_i x_i \preceq \lambda I_{n_p}, \]
we have using our earlier lower bound on $\norm{x}_1$,
\[\mu_{\text{OPT}} \geq \lambda_{\max}\left(\sum_{i \in [d]} P_i x_i\right) \geq \frac{1}{n_p}\Tr\left(\sum_{i \in [d]} P_i x_i\right) \geq \frac{\norm{x}_1}{n_p} \min_{i \in [d]} \Tr(P_i) \geq \frac{\min_{i \in [d]} \Tr(P_i)}{\max_{i \in [d]} \Tr(C_i)} \cdot \frac{n_c}{n_p}. \]
Finally, the last inequality in the lemma statement follows from
\[\frac{\min_{i \in [d]} \Tr(P_i)}{n_p} \ge \frac{\min_{i \in [d]} \lambda_{\max}(P_i)}{n_p},\; \frac{n_c}{\max_{i \in [d]} \Tr(C_i)} \ge \frac{1}{\max_{i \in [d]} \lambda_{\max}(C_i)}.\]
\end{proof}

By combining these bounds, we have that the range $\mu_{\text{upper}}/\mu_{\text{lower}}$ may always be bounded by 
\[n_p \left(\min_{i \in [d]} \frac{\lambda_{\max}(P_i)}{\lambda_{\min}(C_i)}\right)\left(\frac{\max_{i \in [d]} \lambda_{\max}(C_i)}{\min_{i \in [d]} \lambda_{\max}(P_i)}\right)\le n\cdot\frac{\max_{i \in [d]} \lambda_{\max}(C_i)}{\min_{i \in [d]} \lambda_{\min}(C_i)}.\]
In other words, up to a multiplicative $n$, the range is the ratio of the largest eigenvalue amongst all the covering matrices, to the smallest. Further note that the notion of width used in the paper,
\[\rho \defeq \max_{i \in [d]} \frac{\lmax{\tilde{C}_i}}{\lmax{\tilde{P}_i}},\]
may be bounded for the instances passed to the oracle via the reduction in Lemma~\ref{lem:reduction}. In particular, for $\{\tilde{C}_i\}_{i = 1}^d = \{C_i\}_{i = 1}^d, \{\tilde{P}_i\}_{i = 1}^d = \{(1 - \eps/3)\mu^{-1} P_i\}_{i = 1}^d$, we see that the largest value $\mu$ takes is at $\mu = \mu_{\text{upper}}$. This leads to the corresponding notion of width for the original instance \eqref{eq:possdpdef},
\[\rho = \mu_{\text{upper}} \cdot \max_{i \in [d]} \frac{\lambda_{\max}(C_i)}{\lambda_{\max}(P_i)} = \left(\min_{i \in [d]} \frac{\lambda_{\max}(P_i)}{\lambda_{\min}(C_i)}\right)\left(\max_{i \in [d]} \frac{\lambda_{\max}(C_i)}{\lambda_{\max}(P_i)}\right) \le \max_{i \in [d]} \frac{\lambda_{\max}(C_i)}{\lambda_{\min}(C_i)},\]
i.e. the largest condition number amongst all the covering matrices.
\section{Technical lemmata}
\label{app:technical}

Here, we give omitted proofs for technical lemmata used in the main body of the paper.

\restateEntryBound*
\begin{proof}
	The bound $Knd\rho \leq (nd\rho)^5$ comes from the definition of $K$ and our assumed bound on $\eps$. It clearly suffices to show that no entry of any of the $\xcurr_i C_i$ has larger magnitude than $Kn\rho$.
	
	Suppose for contradiction that for $j, k \in [n_c]$, \[\left|\left[\xcurr_i C_i\right]_{jk}\right| \geq Kn\rho.\] Because the $2 \times 2$ restriction of the matrix to indices $j$ and $k$ is positive semidefinite, by considering its quadratic form against the vector $\begin{pmatrix}1 & -1\end{pmatrix}$, we have either
	\[\left[\xcurr_i C_i\right]_{jj} \geq Kn\rho \text{ or } \left[\xcurr_i C_i\right]_{kk} \geq Kn\rho.\]
	Therefore, the trace of this matrix is at least $Kn\rho$, implying its largest eigenvalue is at least $K\rho$. Finally, this implies that the largest eigenvalue of $\xcurr P_i$ is at least $K$, contradicting the fact that the algorithm has not terminated. The same argument can be used to obtain a trace bound (and thus an entry magnitude bound) on the other two matrices in the lemma statement, as this would also imply a trace bound on $\Cxcurr$, as
	\begin{equation*}
	\Tr\left[\Cxcurr\right] \geq \max\left\{\Tr\left[\alpha \sum_{i = 1}^{d} \delta_i^{(t)} x_i^{(t)} C_i\right], \Tr\left[\alpha \sum_{i = 1}^{d} \left(\delta_i^{(t)}\right)^2 x_i^{(t)} C_i\right]\right\}.
	\end{equation*}
	This is immediate from $\alpha \leq 1$ and all $\delta_i^{(t)} \leq 1$, and monotonicity of trace.
\end{proof}

\restateAddGronwall*
\begin{proof}
	Define the function $v(t) = \exp(-\beta t)$, such that $v'(t) = -\beta v(t)$. Then, we have
	\begin{align*}
	\frac{d}{dt} \frac{u(t)}{v(t)} = \frac{u'(t) v(t) - v'(t) u(t)}{v(t)^2} \geq \frac{-\beta u(t) v(t) + \beta u(t) v(t) - cv(t)}{v(t)^2} = -\frac{c}{v(t)}.
	\end{align*}
	Integrating both sides yields
	\begin{equation*}
	\frac{u(t)}{v(t)} - \frac{u(0)}{v(0)} \geq -c\int_0^t \exp(\beta s) ds = -\frac{c}{\beta}\left(\exp(\beta t) - 1\right).
	\end{equation*}
	Using $v(0) = 1$ and $v(t) = \exp(-\beta t)$, we have the first conclusion upon rearrangement. The second conclusion follows from $\exp(-x) \geq 1 - x$ for $|x| < 1$, applied twice.
\end{proof}

\begin{lemma}[$\nu$ bound]
	\label{lemma:nubound}
	For the function $\nu(\beta,\gamma)$ introduced in Lemma \ref{lemma:iotaderiv}, we have $0 \leq \nu(\beta,\gamma) \leq \frac{e^{-\beta} + e^{-\gamma}}{2}$ for any $\gamma, \beta \geq 0$.  
	\end{lemma}
	\begin{proof}
	We restate the function here for convenience:
	\[
	\nu(\beta, \gamma) = \begin{cases}
	\frac{e^{-\beta} - e^{-\gamma}}{\gamma - \beta} &\text{if $\gamma \neq \beta$}\\
	e^{-\gamma} & \text{if $\gamma = \beta$}
	\end{cases}
	\] 
	Observe that the claim is trivially true if $\beta = \gamma$. In addition, we note that our claim is symmetric in $\beta$ and $\gamma$. Therefore we assume $\beta < \gamma$. We start with the lower bound: as $\nu(\beta,\gamma) = \frac{e^{-\beta}- e^{-\gamma}}{\gamma-\beta}$, we observe that both the numerator and denominator are positive. 
	
	For the upper bound, we observe that our claim is equivalent to
	\[
	\frac{e^{-\beta}- e^{-\gamma}}{e^{-\beta}+ e^{-\gamma}} \leq \frac{\gamma-\beta}{2}.
	\]
	Now if $z = \gamma - \beta > 0$ the left hand side of the above equals $\frac{e^z - 1}{e^z+1} = \tanh(z/2)$. The result follows via the standard fact $\tanh(x) < x$ for $x > 0$. 
	\end{proof}

Finally, we give a proof of the Extended Lieb-Thirring inequality via the technique developed in the main body of the paper of encoding inequalities in higher dimensions.

\begin{lemma}[Extended Lieb-Thirring]
	For $A \succ 0, B \succeq 0$ and $\alpha \in [0, 1]$,
	\[\inprod{B^{1/2} A^{\alpha} B^{1/2}}{B^{1/2} A^{1 - \alpha} B^{1/2}} \leq \inprod{B^2}{A}.\]
\end{lemma}
\begin{proof}
Define the function, for $\alpha \in [0, 1]$,
\[f(\alpha) = \inprod{B^{1/2} A^{\alpha} B^{1/2}}{B^{1/2} A^{1 - \alpha} B^{1/2}},\]
and note $f$ is symmetric around $1/2$. We will show that for any $\alpha \in [0, 1/2]$,
\begin{equation}\label{eq:mainclaimalpha}f(0) + f(2\alpha) \ge 2f(\alpha).\end{equation}
To see why this yields the lemma statement, let
\[M = \max_{\alpha \in [0, 1]} \frac{f(\alpha)}{f(0)}.\]
We claim if \eqref{eq:mainclaimalpha} holds, then $M = 1$, proving the lemma statement. In particular, if $M > 1$, then for the maximizing value of $\alpha$, \eqref{eq:mainclaimalpha} implies a tighter bound of $\frac{1+M}{2}$, contradicting the definition of $M$. Now, we prove \eqref{eq:mainclaimalpha}. First, note that for any $\alpha \in [0, 1]$, $\lambda \ge 0$, and scalars $\tilde{u}, \tilde{v}$, by AM-GM,
\begin{equation}\label{eq:amgm}\tilde{u}^2 + \lambda^{2\alpha}\tilde{v}^2 \ge 2 \lambda^{\alpha}\tilde{u}\tilde{v}.\end{equation}
Next, consider the matrix
\[\begin{pmatrix}B^{1/2} I B^{1/2}& -B^{1/2} A^{\alpha} B^{1/2} \\ -B^{1/2} A^{\alpha} B^{1/2} & B^{1/2} A^{2\alpha} B^{1/2}\end{pmatrix}.\]
We claim it is positive semidefinite: taking its quadratic form with respect to vector $\begin{pmatrix} u^\top & v^\top \end{pmatrix}$
and letting $\tilde{u} = B^{1/2} u, \tilde{v} = B^{1/2} v$ gives
\[\tilde{u}^\top A^{\alpha} \tilde{u} - 2\tilde{u}^\top A^{1/2} \tilde{v} + \tilde{v}^\top A^{1 - \alpha} \tilde{v}. \]
By diagonalizing $A$, this follows from \eqref{eq:amgm}. Similarly, the following matrix is also PSD:
\[\begin{pmatrix}B^{1/2} A B^{1/2} & B^{1/2} A^{1 - \alpha} B^{1/2} \\ B^{1/2} A^{1 - \alpha} B^{1/2} & B^{1/2} A^{1 - 2\alpha} B^{1/2}\end{pmatrix}.\]
Taking the (nonnegative) trace product of these two matrices yields the inequality \eqref{eq:mainclaimalpha}.
\end{proof} 	%
\section{Runtime improvements}
\label{app:pure}

\subsection{Pure packing and covering}

In this section, we prove the runtime improvement for pure packing and covering instances for completeness. Consider the case when the instance \eqref{eq:feasible}, \eqref{eq:infeasible} is a pure packing instance, i.e. all of the covering matrices $C_i$ are $1 \times 1$ scalars $c_i$. In this setting, we note that in every iteration $t$, the matrix
\[Z^{(t)} = \exp\left(-\sum_{i = 1}^d \xcurr_i C_i\right)\]
is also a scalar, and therefore each of the gradient coordinates
\[\cGradi = \inprod{C_i}{\frac{Z^{(t)}}{\Tr Z^{(t)}}} = c_i\]
in every iteration, so that $\cGrad = c$ is a constant. Similarly, for pure covering instances, $\pGrad = p$ is a constant. We now appropriately modify Lemma~\ref{lem:numSlowIters}.
\begin{lemma}\label{lem:numSlowItersPack}
For a pure packing instance, the number of iterations is at most $\Oh{\log^2 (nd\rho) /\epsilon^2}$. Moreover, for a pure covering instance, the number of slow iterations over all phases is at most $\Oh{\log^2 (nd\rho) /\epsilon^2}$.
\end{lemma}
\begin{proof}
	First, consider when $\cGrad = c$ is a constant. Label the iterations $1, 2, \ldots T$ for simplicity, and let $i\in [d]$ satisfy \eqref{existenceSmallPG} for all but the last iteration. 
	Because $i$ satisfies \eqref{existenceSmallPG}, $\sum_{t \in [T - 1]} \pGradi \le (1- \eps)(T - 1)c_i$. Thus,
	\begin{align*}
	x_i^{(T)} &\geq x_i^{(1)} \exp\left( \frac{\alpha}{4} \sum_{t\in[T - 1]} \frac{\left(1 - \frac{\eps}{2}\right) c_i - \pGradi}{c_i} \right)\\
	&\geq x_i^{(1)} \exp \left( \frac{\alpha}{8} \eps(T - 1)   \right).
	\end{align*} 
	The proof of Lemma~\ref{lem:numSlowIters} then implies the desired iteration bound. Now, consider the case where $\pGrad = p$ is a constant, and label the (not necessarily consecutive) slow iterations $1, 2, \ldots T$. By definition, for all $t \in [T]$,
	\[p_i \ge \frac{1}{3}\cGradi.\]
	This allows us to modify the bound \eqref{eq:deltabound} to read
	\[\delta_i^{(t)} \ge \frac{\left(1 - \frac{\eps}{2}\right)\cGradi - \pGradi}{6\pGradi},\]
	as either the numerator is nonpositive and this is certainly true, or we use the definition of slow iteration. Now, again recall $i$ satisfying \eqref{existenceSmallPG} implies $(T - 1) p_i \le (1- \eps)\sum_{t \in [T - 1]} \cGradi$. Thus,
	\begin{align*}
	x_i^{(T)} &\geq x_i^{(1)} \exp\left( \frac{\alpha}{12} \sum_{t\in[T - 1]} \frac{\left(1 - \frac{\eps}{2}\right) \cGradi - p_i}{p_i} \right)\\
	&= x_i^{(1)} \exp \left( \frac{\alpha}{12}\left(\frac{1 - \frac{\eps}{2}}{1 - \eps} - 1\right)(T - 1)   \right)\\
	&\ge x_i^{(1)} \exp \left( \frac{\alpha\eps}{36}(T - 1)   \right).
	\end{align*} 
	The conclusion follows again from a bound on $T$. 
\end{proof}

\subsection{Commuting covering matrices}
\label{app:commute}

In this section, we prove that the runtime and parameters of Algorithm~\ref{alg:sdpalg} can be independent of the width parameter $\rho$ in the case when the covering matrices commute. In particular, consider setting the threshold $K$ in the algorithm to $4\log(nd)/\eps$, and modifying all statements in the analysis to be independent of $\rho$. We note that the only place that we explicitly use this dependence on $\rho$ in Section~\ref{sec:sdpfull} is in the proof of Lemma~\ref{lemma:largebucket}, the analysis of terms in the large bucket. To this end, we prove the following improved bound in this particular case.

\begin{lemma}[Large bucket bound, commuting matrices]
	\label{lemma:largebucketimprove}
	Suppose all covering matrices in the set $\allC$ commute for a positive SDP instance. For any $0 \leq t \leq 1$, 
	\[
	\sum_{(i,j) \in L} (\rcovstep_{ij})^2 \nu(\lambda_i,\lambda_j) \leq (n_c d)^{-60}\Tr\exp(-\Psi)
	\]
	and
	\[
	\left| \sum_{(i,j) \in L} \rcovstep_{ij} \rcovsq_{ij} \nu(\lambda_i,\lambda_j) \right| \leq (n_c d)^{-60}\Tr\exp(-\Psi).
	\]
\end{lemma}
\begin{proof}
We note that in the case of commuting covering matrices, under the diagonalization induced by the rotation $Q$, all terms except for the diagonal entries drop out. Recalling the definition $\nu(\lambda_i, \lambda_i) = e^{-\lambda_i}$, we wish to prove
\[\sum_{j: \lambda_j > 4K} \lambda_j^2 \exp(-\lambda_j) \le (n_c d)^{-60} \sum_{j \in [n_c]} \exp(-\lambda_j). \]
Now, note that there is at least one eigenvalue bounded above by $K$ by the termination condition. Therefore, it suffices to show that for any $\lambda > 4K$,
\[\lambda^2 \exp(-(\lambda - K)) \le (n_c d)^{-100},\]
and then sum over all the $\lambda_j > 4K$ after multiplying by $\exp(K)$. To this end, we note that $\lambda^2\exp(-(\lambda - K))$ is a decreasing function of $\lambda$ for $\lambda$ large, and at $\lambda = 4K$ the inequality holds.
\end{proof}

Overall, this improvement carries through the rest of the analysis, as nowhere else did we use an additional dependence on $\rho$. 	%
\section{Analysis under approximate guarantees}
\label{app:approximate}

In this section, we show how our algorithm is tolerant to various types of approximation errors, which are incurred from the more-efficient implementations of Sections~\ref{ssec:termination} and~\ref{ssec:JL}. In particular, we will show tolerance to the following types of error.
\begin{enumerate}
	\item Terminating the algorithm with a threshold in the range $[K, 2K]$, in line 4 of Algorithm~\ref{alg:sdpalg}.
	\item Computing  $\pGradi$ within a $1 \pm \frac{\eps}{20}$ multiplicative factor, in line 7 of Algorithm~\ref{alg:sdpalg}.
	\item Computing $\cGradi$ within a $1 \pm \frac{\eps}{20}$ multiplicative factor plus an $\exp(-10K) \Tr[C_i]$ additive factor in line 8 of Algorithm~\ref{alg:sdpalg}.
\end{enumerate}
We first discuss the effects of late termination on the analysis. The proof of Lemma~\ref{lem:terminationimpliesfeasibility} remains the same, as termination still implies a threshold larger than $K$, which is all the analysis requires. Moreover, in the iteration bounds of Section~\ref{ssec:iteration}, the number of possible phases does not increase by more than a constant factor (Lemma~\ref{lem:numphases}), and the same constant-factor increase holds due to late termination in Lemmas~\ref{lem:numfastiters} and~\ref{lem:numSlowIters}. This does not affect any certificates of infeasibility.

Next, we discuss the effects of approximate computation of the gradients. Denote our approximations to the true gradients $\pGradi$, $\cGradi$ in an iteration by
\[\widetilde{\pGradi} \in \left[ \left(1 - \frac{\eps}{20}\right) \pGradi,\; \left(1 + \frac{\eps}{20}\right) \pGradi \right],\; \widetilde{\cGradi} \in \left[\left(1 - \frac{\eps}{20}\right) \cGradi,\; \left(1 + \frac{\eps}{20}\right)\cGradi + e^{-10K} \Tr[C_i] \right].\]

By the termination guarantee, throughout the algorithm $I_{n_p} \preceq \exp\left(\Pxcurr\right) \preceq e^{2K} I_{n_p}$. Thus,
\begin{equation}\label{eq:nondegenerate}\pGradi = \frac{\inprod{P_i}{\exp\left(\Pxcurr\right)}}{\Tr\exp\left(\Pxcurr\right)} \geq \frac{e^{-2K}}{n} \Tr[P_i] \geq \frac{e^{-2K}}{n^2 \rho} \Tr[C_i] \ge e^{-3K}\Tr[C_i].\end{equation}
The second inequality used $\Tr[C_i] \le n\lmax{C_i} \le n\rho\Tr[P_i]$. Note that computing the update set $W^{(t)}$ with respect to these approximations allow us to guarantee that for every $i \in W^{(t)}$,
\begin{align*}\widetilde{\pGradi} \le \left(1 - \frac{\eps}{2}\right)\widetilde{\cGradi} &\Rightarrow \left(1 - \frac{\eps}{20}\right)\pGradi \le \left(1 - \frac{\eps}{2}\right)\left(\left(1 + \frac{\eps}{20}\right)\cGradi + e^{-10K}\Tr[C_i]\right) \Rightarrow \pGradi \le \left(1 - \frac{\eps}{3}\right)\cGradi.\end{align*}
In the last implication, we used \eqref{eq:nondegenerate} to bound the effect of the additive term. Conversely, whenever the update set is determined to be empty, Lemmas~\ref{lem-sdp-infeas} and~\ref{lem-coord-update} still certify infeasibility of \eqref{eq:infeasible} against the threshold $1 - \eps$. Because all $\widetilde{\delta^{(t)}_i}$ are still bounded in the range $[0, \half]$, the analysis of the potential bounds in Section~\ref{sec:sdpfull}, i.e. Lemmas~\ref{lemma:packing} and Lemma~\ref{lemma:covering} hold. In order to show Theorem~\ref{thm:mainclaimsdp}, we modify the definition of the potential function to be
\[f(x^{(t)}) \defeq (1 - \eps)^2\smax\left(\Pxcurr\right) - (1 + \eps)\smin\left(-\Cxcurr\right).\]
In particular, the remaining step in the proof of Theorem~\ref{thm:mainclaimsdp} is to show, in the vein of \eqref{eq:keyclaim},
\[(1 - \eps)\left(1+ \widetilde{\delta^{(t)}_i}\right)\pGradi - (1 + \eps)\left(1- \widetilde{\delta^{(t)}_i}\right)\cGradi \le \left(1+ \widetilde{\delta^{(t)}_i}\right)\widetilde{\pGradi}- \left(1- \widetilde{\delta^{(t)}_i}\right)\widetilde{\cGradi} + (nd\rho)^{-15} \le (nd\rho)^{-15}.\]
The first inequality is due to the approximation guarantee, where without loss of generality we scale the matrices $\{C_i\}_{i \in [d]}$ to have polynomially bounded trace, and the second is due to the proof of Theorem~\ref{thm:mainclaimsdp}. Thus, the result is that this potential function does not increase more than a negligible $2(nd\rho)^{-15}$. Finally, the guarantee of Lemma~\ref{lem:terminationimpliesfeasibility} does not change by more than a constant multiple of $\eps$, and the iteration bounds also are not affected by more than constant factors because of the approximation guarantee on the step size. 	\end{appendix}
	
\end{document}